\documentclass[12pt,a4paper]{article}
\usepackage{amsmath,amssymb,graphicx,tocloft,enumerate,amsthm,etoolbox}
\usepackage[top=1in, bottom=1in, left=1in, right=1in]{geometry}
\usepackage{mathtools}
\usepackage{xcolor}
\usepackage{cite}
\usepackage{multirow}
\usepackage{tabularx}
\usepackage{physics}
\usepackage{enumitem}

\patchcmd\subequations
 {\theparentequation\alph{equation}}
 {\subequationsformat}
 {}{}

\newcommand{\subequationsformat}{\theparentequation.\arabic{equation}}

\newcommand{\e}{\mathrm{e}}
\newtheorem{definition}{\textsc{Definition}}[section]
\newtheorem{theorem}{\textsc{Theorem}}[section]

\newtheorem{lemma}{\textsc{Lemma}}[section]
\newtheorem{remark}{\textsc{Remark}}[section]

\usepackage{titlesec}

\titleformat{\paragraph}
{\normalfont\normalsize\bfseries}{\theparagraph}{1em}{}
\titlespacing*{\paragraph}
{0pt}{3.25ex plus 1ex minus .2ex}{1.5ex plus .2ex}

\allowdisplaybreaks
\begin{document}

\begin{center}
{\Large All meromorphic solutions of a 3D Lotka-Volterra system: detecting partial integrability}
\vskip 3mm
Techheang Meng\footnote{techheang.meng.21@ucl.ac.uk} and Rod Halburd\footnote{r.halburd@ucl.ac.uk}
\vskip 3 mm
Department of Mathematics, University College London, \\ Gower Street, London WC1E 6BT
\end{center}

\begin{abstract}
For an autonomous system of ordinary differential equations, the existence of a meromorphic general solution is equivalent to the Painlev\'e property, which is widely used to detect integrability. We find all meromorphic solutions of a multi-parameter three-dimensional Lotka-Volterra system.  Some cases correspond to particular choices of the parameters for which only some solutions are meromorphic, while the general solution is branched.  The main difficulty is to prove that all meromorphic solutions have been found.  The proof relies on a detailed study of local series expansions combined with value distribution results from Nevanlinna theory.
\end{abstract}

\section{Introduction}
\label{introduction}
In 1889, Kowalevskaya \cite{kowalong} famously found all choices of parameters in the equations of motion for a spinning top for which the general solution, considered over the complex numbers, is meromorphic.  The result included the previously known integrable symmetric top together with two other known integrable cases due to Euler and Lagrange, as well as one new case, now called the Kowalevskaya top, which she was able to integrate \cite{kowashort} using Riemann theta functions.  In fact, in an earlier letter to Mittag-Leffler from December 1884 (see \cite{KM} pages 80--82), Kowalevskaya described unpublished work in which she analysed the system of equations
\begin{equation}
\label{kowalevskaya-LV}
	\dot\omega_j=\omega_j\sum_{k=1}^3 a_{jk}\omega_k,\quad j=1,2,3,
\end{equation}
where $a_{jk}$ are constants.  She observed that if $a_{12}a_{23}a_{31}=a_{13}a_{32}a_{21}$, then there is a meromorphic solution depending on three arbitrary parameters such that all poles are simple.  Kowalevskaya commented that the general solution can be expressed explicitly in terms of elliptic functions.

The idea that a meromorphic general solution suggests that an ODE is in some sense integrable was extended by Painlev\'e and his school to non-autonomous equations.  Solutions of non-autonomous equations can have fixed as well as movable singularities.  A fixed singularity of a solution is one that occurs at a value of the independent variable at which the equation itself, not just the solution, is singular.  The location of movable singularities depend on initial conditions.  An ODE is said to possess the {\em Painlev\'e property} if all solutions are single-valued about all movable singularities.

In this paper, we will study the following third-order Lotka-Volterra system:
\begin{subequations}
\label{eq:1}
	\begin{align}
    \label{eq:1,1}
	x'=x(Cy+z+\lambda),\\
    \label{eq:1,2}
	y'=y(Az+x+\mu),\\
    \label{eq:1,3}
	z'=z(Bx+y+\nu),
	\end{align}
\end{subequations}
where $A$, $B$, $C$, $\lambda$, $\mu$ and $\nu$ are complex parameters.  The case $\lambda=\mu=\nu=0$ is a special case of the system
(\ref{kowalevskaya-LV}) studied by Kowalevskaya, and is often referred to as the $ABC$ Lotka-Volterra system.

The Painlev\'e property for the Lotka-Volterra system (\ref{eq:1}) was analysed in \cite{bountis84,lvsystem90}.
In this paper, we will determine all meromorphic solutions of the system  (\ref{eq:1}), even in cases for which the general solution is not meromorphic. Hence we provide another example of extending the ideas of Kowalevskaya and Painlev\'e to partially integrable systems in the sense that solutions with nice singularity structure can be characterised (in our case, found explicitly), even when the equation is not fully integrable.  The main challenge is to show that all meromorphic solutions have been found.   Some meromorphic solutions of (\ref{eq:1}) and some general solutions (corresponding to the integrable cases) have already been found (see \cite{lvsystem90}).

The starting point for Kowalevskaya's analysis, and what has become known as the Painlev\'e test, is to look for sufficiently general Laurent series expansions of solutions.  In many cases, an obstruction to such an expansion can be used to prove the existence of a branched solution.  When looking for equations with the Painlev\'e property, an equation can be discarded as soon as branching around a movable singularity is discovered in any solution. 

A more subtle analysis is required to prove that branching occurs in all solutions other than some particular solutions. Nevanlinna theory \cite{hayman1964meromorphic}, which describes the value distribution of meromorphic functions, provides the extra global information required.  Nevanlinna theory has many applications to differential equations \cite{laine93}.  One important such application is to show that all meromorphic solutions (or sometimes all {\em admissible} meromorphic solutions, namely all meromorphic solutions that are in some sense more complicated than the coefficients in an equation) must take a particular value many times. If this value is a singularity of the equation, then by studying the appropriate Laurent series expansion, necessary conditions can then be deduced that must apply to all such solutions. Another application concerns equations with the {\em finiteness property} defined by Eremenko in \cite{eremenkoln} but used in several previous works \cite{hille:78,eremenko1986meromorphic,eremenko2005}, which applies to certain autonomous equations for which any expansion of a solution around a pole takes one of only a finite number of forms. In such cases, the Laurent series expansions contain no free parameters.  This property does not apply in many of the cases that arise in our analysis of the system (\ref{eq:1}) as the expansions have non-negative integer resonances.  

The Hayman equation $ww''-(w')^2=\alpha(z)w +\beta(z) w'+\gamma(z)$, where $\alpha$, $\beta$ and $\gamma$ are meromorphic functions, is a simple ODE for which the resonances in an expansion about a zero of $w$ can occur at arbitrary values determined by the coefficients.  If the resonances are non-integer, then the equation has the finiteness property, but even in the constant coefficient case, resonances can be arbitrarily large integers.  When $\alpha$, $\beta$ and $\gamma$ are constants, all meromorphic solutions were found in Chiang and Halburd \cite{chiang2003meromorphic}. In the case of general meromorphic coefficients, all admissible meromorphic solutions were found in Halburd and Wang \cite{halburdwang14}.  The main idea is to show that, provided the resonance does not occur in the first couple of terms in the Taylor series expansion, a rational function of $w$ and $w'$ can be constructed that, as a meromorphic function in $z$, is ``small'' in the sense of Nevanlinna compared to $w$.  This function can then be determined explicitly in terms of the coefficients and possibly an arbitrary constant, showing that $w$ must be an admissible solution of a first-order equation, which can be solved.

Some of these ideas are key to our analysis, however the problem of finding all meromorphic solutions of (\ref{eq:1}) is novel in several ways.  The first is that it is a system of equations rather than a single scalar equation, but more importantly, it is a system with a number of different leading-order behaviours, often generating Laurent series expansions with non-negative integer resonances.  In general, our analysis needs to vary depending on which of these singularities is present.  

The Lotka-Volterra system (\ref{eq:1}) is non-singular for any initial condition where $x$, $y$ and $z$ are all finite. So constraints arising from the existence of a Laurent series expansion will only arise when at least one of the dependent variables has a pole.
To this end, we look for formal series solutions of the form
\begin{equation}
\label{eq:2}
\begin{aligned}
    x=\sum_{k=0}^\infty x_k(t-t_0)^{p_x+k},\quad y=\sum_{k=0}^\infty y_k(t-t_0)^{p_y+k},\quad z=\sum_{k=0}^\infty z_k(t-t_0)^{p_z+k},
\end{aligned}    
\end{equation}
where at least one of $p_x,p_y,p_z$ is negative and $x_0y_0z_0\neq0$. It has been shown in \cite{lvsystem90} that
\begin{equation}
\label{orderpole}
(p_x,p_y,p_y)=(-1,-1,-1),\ (-1,-1,\gamma),\ (\alpha,-1,-1)\mbox{\ \ or\ \ }(-1,\beta,-1),
\end{equation}
where $\alpha,\beta,\gamma\in \mathbb{Z}_0^+$. We refer to these cases respectively as the $\mathbf{p_0}$ poles, $\mathbf{p_z}$ poles, $\mathbf{p_x}$ poles, and $\mathbf{p_y}$ poles. Here, $\mathbb{Z}_0^+$ is the set of non-negative integers. We will denote  the $x_k$ term in (\ref{eq:2}) by  $x_k,\hat{x}_k$, $\breve{x}_k$, $\check{x}_k$ for $\mathbf{p_0}$, $\mathbf{p_z}$, $\mathbf{p_x}$, $\mathbf{p_y}$ poles respectively. Further details of the coefficients of (\ref{eq:2}) and other poles will be given in section \ref{painleveana}. We will see that the discussion of meromorphic solutions of (\ref{eq:1}) depends on whether $ABC+1$ vanishes and on the relationships between $\lambda,\mu,$ and $\nu$. Accordingly, we divide the analysis into three sections (sections \ref{holosolu}-\ref{abc+1not0}). The main result of this paper is the following theorem.
\begin{theorem}
\label{mainthm}
Let $D=ABC+1$, $\mathrm{j}=-(1+\mathrm{i}\sqrt{3})/2$ and $K,L\in\mathbb{C}$. Then all meromorphic solutions of (\ref{eq:1}) are listed in Table \ref{tab1} or can be obtained from them by applying one of the following transformations:
\begin{subequations}
\label{transmaps}
\begin{align}
    \label{cyclemaps1}
    &\pi_1(x,y,z,B,C,A,\lambda,\mu,\nu)=(y,z,x,C,A,B,\mu,\nu,\lambda),\\
    \label{cyclemaps2}
    &\pi_2(x,y,z,B,C,A,\lambda,\mu,\nu)=(z,x,y,A,B,C,\nu,\lambda,\mu),\\
    \label{swapmaps1}
    &\pi_{xy}(x,y,z,B,C,A,\lambda,\mu,\nu)=(Cy,Bx,Az,1/C,1/B,1/A,\mu,\lambda,\nu),\\
    \label{swapmaps2}
    &\pi_{yz}(x,y,z,B,C,A,\lambda,\mu,\nu)=(Bx,Az,Cy,1/B,1/A,1/C,\lambda,\nu,\mu),\\
    \label{swapmaps3}
    &\pi_{zx}(x,y,z,B,C,A,\lambda,\mu,\nu)=(Az,Cy,Bx,1/A,1/C,1/B,\nu,\mu,\lambda),\\
    \label{conjugatemap}
    &\pi_{c}(x,y,z,B,C,A,\lambda,\mu,\nu)=(\bar{x},\bar{y},\bar{z},\bar{B},\bar{C},\bar{A},\bar{\lambda},\bar{\mu},\bar{\nu}) ,
    \end{align}
\end{subequations}
where $\bar{\lambda}$ is the conjugate of $\lambda$ and $\bar{f}$ is defined by $\bar f(z)=\overline{f(\bar z)}$.
The first column in Table \ref{tab1} lists the types of singularities that $x,y,$ and $z$ must admit in each case, with no other poles present.
The 2D L-V poles in the first column in Table \ref{tab1} corresponds to the two-dimensional L-V system.
\begin{table}[ht]
\caption{Canonical meromorphic solutions of (\ref{eq:1})}
\label{tab1}
\vspace{-0.5cm}
\begin{center}
\begin{tabular}{|m{4.1em}|m{2.8em}|m{27em}|}
     \hline
      \text{Poles} & \text{Section} & \text{Meromorphic solutions}  \\    \hline
      \multirow{2}{5em}{None} & \multirow{2}{1.5em}{\ref{holosolu}} & $\{x,y,z\in\mathbb{C}\},\{x=y=0,z=K\e^{\nu t}\}$,$\{A=x=0,z'=z(y+\nu),$\\
      & & $y'=\mu y\}$,$\{D=0,\lambda=\mu=\nu,z=-Cy=BCx=BCK\e^{\lambda t}\}$\\
      \hline
      $\mathbf{p_0}$ & \ref{0-3pole} & $D=0,B=(C-1)/C,$ (\ref{0-3poleeq:01}),(\ref{0-3poleeq:02}),(\ref{0-3poleeq:1})\\
      \hline
      \multirow{2}{5em}{$\mathbf{p_z}$} & \multirow{2}{1.5em}{\ref{0-x-y=0}} & $D=\gamma=0,\lambda=\mu,\{x'=x(x+L\e^{\nu t}+\lambda),y=C^{-1}x,z=L\e^{\nu t}\}$,\\
      & & $\{y=C\nu(Ke^{-C\nu t}-C)^{-1},x=K\e^{-C\nu t}y, z=-\lambda,\lambda\neq0\}$\\
      \hline
      2D L-V & \ref{0-x-y=0} & $\lambda=\mu,z=0,x'=x(x+K\e^{\lambda t}+\lambda),y=C^{-1}(x+K\e^{\lambda t})$\\
      \hline
      \multirow{2}{5em}{$\mathbf{p_z},\mathbf{p_x}$} & \multirow{2}{1.5em}{\ref{0-x-y-y-z==0}} & $A=C^{-1}=\gamma+1,B=1,\lambda=\mu=\nu,x'=x(x+K\e^{\lambda t}+\lambda),$\\
      & & $z=(K\e^{\lambda t}+x)(CL\e^{-A\lambda t}x^{A}+1)^{-1},y=L\e^{-A\lambda t}x^{A}z$\\
      \hline
      \text{$\mathbf{p_z},\mathbf{p_x},\mathbf{p_y}$} & \ref{0-2pole} & $D=0,\lambda=\mu=\nu$, see section \ref{0-2pole}\\
      \hline
      $\mathbf{p_0}$ & \ref{neq03*} & $D\neq0$, (\ref{neq03*eq:1}),(\ref{neq03*eq:2.1}),(\ref{neq03*eq:2.2}),(\ref{neq03*eq:2.3})\\
      \hline
      \multirow{2}{5em}{$\mathbf{p_z}$} & \multirow{2}{1.5em}{\ref{neq0x-y}} & $B=-C^{-1},\lambda\neq\mu,\nu(\mu-A\lambda)=0,z=(\mu-\lambda)(1-A)^{-1},$\\
      & & $y'=y(Cy+C\nu+(\mu-A\lambda)(1-A)^{-1}),x=C(y+\nu)$\\
      \hline
      $\mathbf{p_z},\mathbf{p_0}$ & \ref{neq0x-y3} & $D\neq0,B=-C^{-1}-\gamma,$(\ref{3x-y3eq:22}),(\ref{3x-y3eq:25}),(\ref{3x-y3eq:28,31}),(\ref{l=mx-y3eq:1}),(\ref{l=mx-y3eq:4}),(\ref{l=mx-y3eq:6})\\
      \hline
      \text{$\mathbf{p_z},\mathbf{p_x}$} & \ref{neq0x-yy-z} & $D\neq0$, (\ref{l=my-zz-zeq:5}),(\ref{l=my-zz-zsubeq:5.1}),(\ref{l=my-zz-zsubeq:5.2}),(\ref{l=m=nx-yy-zeq:5})\\
      \hline
      \text{$\mathbf{p_z},\mathbf{p_x},\mathbf{p_0}$} & \ref{l=mx-yy-z3} & $D\neq0,\lambda=\mu\neq\nu,A=-C^{-1}=B+\gamma$,(\ref{l=mx-yy-z3eq:11}),(\ref{l=mx-yy-z3eq:12}),(\ref{l=mx-yy-z3eq:13}),(\ref{l=mx-yy-z3eq:16})\\
      \hline
      \text{$\mathbf{p_z},\mathbf{p_x},\mathbf{p_y}$} & \ref{l=mx-yy-zz-x} & $\lambda=\mu\neq\nu,B=C=-A^{-1},C^2+C\gamma+1=0$ (\ref{l=mx-yy-zz-xeq:11}),(\ref{l=mx-yy-zz-xeq:20}),(\ref{l=mx-yy-zz-xeq:24}) \\
      \hline
      \text{$\mathbf{p_z},\mathbf{p_x},\mathbf{p_y}$} & \ref{l=m=nx-yy-zz-x} & $D\neq0,\lambda=\mu=\nu,$ see sections \ref{neq0x-yy-z} and \ref{l=m=nx-yy-zz-x}\\
      \hline
\end{tabular}    
\end{center}
\vspace{-0.5cm}
\end{table}
\end{theorem}

All the meromorphic solutions of (\ref{eq:1}) have explicit forms as described in Theorem \ref{mainthm}. The proof of Theorem \ref{mainthm} is provided in sections \ref{holosolu}-\ref{abc+1not0}.

\section{Local series analysis}
\label{painleveana}
The expansions described in this section have been obtained previously in \cite[§7]{lvsystem90}.  However, it is important to recall that we are not only concerned with the cases in which the Lotka-Volterra system has the Painlev\'e property and, as such, we also need to consider expansions with fewer than three arbitrary parameters.

If at least one of $x$, $y$, $z$ has a pole at $t=t_0$, then they have Laurent series expansions of the form (\ref{eq:2}), where the possible values of $p_x$, $p_y$, $p_z$ are given by
(\ref{orderpole}).
We begin by trying to develop the Laurent series expansion in the case $(p_x,p_y,p_z)=(-1,-1,-1)$.  
Equation (\ref{eq:2}) now takes the form
\begin{equation}
    \label{3x-y3eq:1}
    {\begin{aligned}
        x=\sum_{k=0}^\infty x_k(t-t_0)^{k-1},\quad y=\sum_{k=0}^\infty y_k(t-t_0)^{k-1},\quad z=\sum_{k=0}^\infty z_k(t-t_0)^{k-1}.
    \end{aligned}}
\end{equation}
Substituting (\ref{3x-y3eq:1}) into (\ref{eq:1}) and equating coefficients of $(t-t_0)^{-2}$ yields
\begin{equation}
\label{eq:3}
    \begin{pmatrix}
    0 & C & 1\\
    1 & 0 & A\\
    B & 1 & 0
    \end{pmatrix}
    \begin{pmatrix}
    x_0\\y_0\\z_0
    \end{pmatrix}
    =\begin{pmatrix}
    -1\\-1\\-1
    \end{pmatrix},
\end{equation}
and for $k\ge 1$, equating coefficients of $(t-t_0)^{k-2}$ yields
\begin{equation}
\label{eq:5}
    \begin{pmatrix}
    k & -Cx_0 & -x_0\\
    -y_0 & k & -Ay_0\\
    -Bz_0 & -z_0 & k
    \end{pmatrix}
    \begin{pmatrix}
    x_k\\y_k\\z_k
    \end{pmatrix}
    =\begin{pmatrix}
    P_k\\Q_k\\R_k
    \end{pmatrix},
\end{equation}
where $P_k,Q_k,R_k$ are polynomials in lower-indexed terms.
The determinant of the matrix in (\ref{eq:3}) is $D=ABC+1$. If $D\neq0$, then the leading-order coefficients are uniquely determined:
\begin{equation}
\label{eq:4}
    (x_0,y_0,z_0)=\left(\frac{-CA+A-1}{D},\frac{-AB+B-1}{D},\frac{-BC+C-1}{D}\right).
\end{equation}
If $D=0$, then there is a resonance condition at leading-order, which we see from equation (\ref{eq:3}) is
\begin{equation}
    \label{eq:8}
    BA-B+1=CB-C+1=AC-A+1=0,
\end{equation} 
where we have used the fact that $ABC+1=0$.

The determinant of the coefficient matrix in equation (\ref{eq:5}) is 
\begin{equation}
\label{eq:7}
    \Delta(k)=(k+1)\left(k^2-k-Dx_0y_0z_0\right).
\end{equation}
The roots of (\ref{eq:7}) are called resonances.  If $k$ is a positive integer resonance then equation (\ref{eq:5}) shows that some linear combination of $P_k$, $Q_k$ and $R_k$ must vanish, which is known as a resonance condition.  If this condition is not satisfied, there is no Laurent series of the required form. If it is satisfied then equation (\ref{eq:5}) with this value of $k$ cannot be used to determine $x_k$, $y_k$ and $z_k$ uniquely in terms of lower-indexed coefficients.
This corresponds to an undetermined parameter in the series expansion at this level. Note that this parameter can sometimes be determined by a later resonance condition.
We see that $k=0$ is a resonance if and only if $D=0$, in which case the resonance condition is given by equation (\ref{eq:8}).
The resonance at $k=-1$ corresponds to the freedom in the choice of the parameter $t_0$.  
Whenever $\Delta(k)\ne 0$ for some positive integer $k$, then $x_k$, $y_k$ and $z_k$ are uniquely determined in terms of lower-indexed coefficients. If $\Delta(k)\ne0$ for all non-negative integers $k$, then the Laurent series expansion is unique.

If $D=0$, the resonances are $\pm1,0$ where (\ref{eq:8}) and $B\lambda x_0-BC\mu y_0-\nu z_0=0$ are the resonance conditions for $k=0,1$ respectively.
If $D\neq0$, then the resonances are $-1,k_1,k_2$ where $k_1+k_2=1$ and $k_1k_2=-Dx_0y_0z_0$. Hence there is at most one positive integer resonance, which must be greater than $1$. 
If a positive integer resonance occurs at $s\geq2$, then $x_i,y_i,z_i,$ for $i<s$ are known, and exactly one of $x_s,y_s,z_s$ is undetermined, say $x_s$. We will refer to this $x_s$ as a parameter. So $x_k,y_k,z_k,$ for $k\geq s$ can be expressed in terms of $x_s$. Once $x_s$ is known, then $x_k,y_k,z_k$ for $k\geq0$ are determined. Moreover, if there is no positive integer resonance, then $x,y,z$ have unique expansions about $\mathbf{p_0}$ poles.

Now consider the case $(p_x,p_y,p_z)=(-1,-1,\gamma)$. The cases $(p_x,p_y,p_z)=$ $(\alpha,-1,-1)$ and $(-1,\beta,-1)$ can be obtained similarly. About $\mathbf{p_z}$ poles, we have
\begin{equation}
    \label{3x-y3eq:3}
    {\begin{aligned}
        x=\sum_{k=0}^\infty \hat{x}_k(t-t_0)^{k-1},\quad y=\sum_{k=0}^\infty \hat{y}_k(t-t_0)^{k-1},\quad z=\sum_{k=0}^\infty \hat{z}_k(t-t_0)^{k+\gamma},
    \end{aligned}}
\end{equation}
where $\hat{x}_0\hat{y}_0\hat{z}_0\neq0$.
A similar argument as in the case $(-1,-1,-1)$ yields the resonances at $k=-1,0,1$. Thus, $\hat{x}_k,\hat{y}_k,\hat{z}_k,$ for $k\geq2$ can be expressed uniquely in terms of $\hat{x}_i,\hat{y}_i,\hat{z}_i$ for $i\leq k-1$. When $k=0$, we have that $C\ne 0$ and
        \begin{equation}
        \label{eq:13}
        C\hat{y}_0=-1,\quad \hat{x}_0=-1,\quad -B-C^{-1}=\gamma,
        \end{equation}
where $-B-1/C=\gamma$ is the resonance condition, and $\hat{z}_0$ is a parameter.
When $k=1$,
\begin{itemize}
    \item if $\gamma>0$ or $\gamma=0=D$, then the resonance condition is
    \begin{equation}
    \label{eq:13.1}
    \lambda=\mu
    \end{equation}
    and exactly one of $\hat{x}_1,\hat{y}_1,\hat{z}_1$ must be a parameter, say $\hat{x}_1$. Thus, $\hat{x}_k,\hat{y}_k,\hat{z}_k, $ for $k\geq1$ can be expressed in terms of $\hat{x}_1,\hat{z}_0$.
    \item if $\gamma=0$ and $D\neq0$, the resonance condition is
    \begin{equation}
    \label{eq:13.2}
    \hat{z}_0=\frac{\mu-\lambda}{1-A},\quad \lambda\neq\mu\mbox{\ \ where\ \ }A\ne 1.
    \end{equation}
    Again, exactly one of $\hat{x}_1,\hat{y}_1,\hat{z}_1$ must be a parameter, say $\hat{x}_1$. Thus, $\hat{x}_k,\hat{y}_k,\hat{z}_k,$ for $k\geq1$ can be expressed in terms of $\hat{x}_1$.
\end{itemize}
\begin{remark}
\label{coexistsabc1=0}
When $D=0$, it follows from (\ref{eq:8}) and (\ref{eq:13}) that $x,y,z$ cannot have $\mathbf{p_0}$ poles if $\mathbf{p_z}$ poles exists (analogously, with the $\mathbf{p_x},\mathbf{p_y}$ poles).  
\end{remark}
\begin{remark}
\label{remarklinearatleast3}
If $a_xx+a_yy+a_zz$ is entire for some $a_x,a_y,a_z\in\mathbb{C}$ not all zero, then there are at most two types of singularities (the proof is just checking the leading order terms).    
\end{remark}
\begin{remark}
\label{zerosxyz}
If $x(t_0)=0$, then either $x=0$ or $t_0$ is a $\mathbf{p_x}$ pole (similarly, for $y,z$). 
\end{remark}
\begin{proof}
Assuming that $t_0$ is not a $\mathbf{p_x}$ pole, then $x,y,z$ are regular at $t_0$. Repeatedly differentiating (\ref{eq:1,1}) yields $x^{(k)}(t_0)=0$ for $k\geq0$ i.e. $x=0$.    
\end{proof}

\section{Main lemmas}
\label{nevanlinna}
We need some results from Nevanlinna theory, which we state without proofs \cite{laine93,hayman1964meromorphic}. 
\begin{definition}
\label{defnevan}
    Let $f$ be meromorphic function, and $n(r,f)$ be the number of poles of $f$ counted with multiplicities inside $\{t\in\mathbb{C}:|t|\leq r\}$. We define the folloowing:
    \begin{align*}
    &\text{Nevanlinna counting function:} &  &N(r,f)=\int_0^r\frac{n(s,f)-n(0,f)}{s}\dd s+n(0,f)\ln r,\\
    &\text{Proximity function:} & &m(r,f)=\frac{1}{2\pi}\int_0^{2\pi}\max\{0,\ln|f(r\e^{i\theta})|\}\dd\theta,\\
    &\text{Nevanlinna charateristic:} & &T(r,f)=m(r,f)+N(r,f).
    \end{align*}
\end{definition}
\begin{definition}
\label{defsmall}{
Let $G(r)$ be a non-negative function such that $G(r)=o(T(r,f))$ as $r\to\infty$ outside a set of finite Lebesgue measure. We denote $G(r)=S(r,f)$. Any such function $G$ is called a small function with respect to $f$ in the Nevanlinna sense.  We also call a meromorphic function $g$ small with respect to $f$ if $T(r,g)$ is small with respect to $f$}
\end{definition}
\begin{definition}
\label{defasymp}
For functions $P_1(r),P_2(r)$, if there exists $\omega_1,\omega_2>0$ satisfying $\omega_1 P_1\leq P_2\leq\omega_2 P_1$ for all large $r>0$ outside a finite Lebesgue measure set, we write $P_1\asymp P_2$. 
\end{definition}
\begin{theorem}
\label{firstfund}(Nevanlinna's First Main Theorem)
    Let $a\in\mathbb{C}$, then 
    \[T(r,f)=T\left(r,1/(f-a)\right)+O(1)\quad
    \mbox{\ as\ \ \ }r\to\infty.
    \]
\end{theorem}
\begin{theorem}
\label{th:1}
Let $U$ be $N$, $m$ or $T$ in Definition \ref{defnevan}, and let $f,g$ be meromorphic.  Then $U(r,fg)\leq U(r,f)+U(r,g)+O(1),~ U(r,f+g)\leq U(r,f)+U(r,g)+O(1)$.
\end{theorem}
\begin{theorem}
\label{th:rational}
    Let $f$ be meromorphic. Then $f$ is transcendental (i.e., non-rational) iff $T(r,f)/\ln r\to\infty$.
\end{theorem}
\begin{theorem}
\label{th:2}
Let $f$ be a non-constant meromorphic function. Then  $m\left(r,f^{(k)}/f\right)=S(r,f)$ for all $k\in\mathbb{N}$ and hence  $T(r,f^{(n)})=O(T(r,f))$.
\end{theorem}
\begin{definition}
\label{defdiffpoly}
Let $f$ be a meromorphic function. We define a differential polynomial as $P(t,f)=\sum_{\lambda\in I}a_\lambda f^{\lambda_0}(f')^{\lambda_1}\cdots(f^{n_\lambda})^{\lambda_{n_\lambda}}$ where $a_\lambda$ are meromorphic functions and $I$ is a finite set. We call $\Lambda=\max_{\lambda\in I}\{\lambda_0+\cdots+\lambda_{n_\lambda}\}$ the total degree of $P$. The term $a_\lambda f^{\lambda_0}(f')^{\lambda_1}\cdots(f^{n_\lambda})^{\lambda_{n_\lambda}}$ is said to be dominant if $\Lambda=\lambda_0+\cdots+\lambda_{n_\lambda}$.
\end{definition}
\begin{theorem}[Clunie's Lemma]
\label{th:4}
Let $f$ be a non-constant meromorphic function. Let $P(t,f)=\sum_{\lambda\in I}a_\lambda f^{\lambda_0}\cdots(f^{n_\lambda})^{\lambda_{n_\lambda}},Q(t,f)=\sum_{\omega\in J}b_\omega f^{\omega_0}\cdots(f^{n_\omega})^{\omega_{n_\omega}}$ be differential polynomials. Suppose that $Q$ has the total degree $\Lambda$. If $f$ satisfies $f^nP(t,f)=Q(t,f)$ such that $T(r,a_\lambda)=S(r,f),T(r,b_{\Lambda})=S(r,f)$, and $\Lambda\leq n$, then $m(r,P(t,f))=S(r,f)$.
\end{theorem}
\begin{theorem}
\label{th:5}
Let $R(f)=P(f)/Q(f)$ where $P(f)=a_0+a_1f+\cdots+a_pf^p$ and $Q(f)=b_0+b_1f+\cdots+b_qf^q$ such that $\gcd(P,Q)=1$. If $T(r,a_i)=S(r,f)$ and $T(r,b_j)=S(r,f)$ for $i=0,\dots,p$ and $j=0,\dots,q$, then $T(r,R(f))=\max\{p,q\}T(r,f)+S(r,f)$.
\end{theorem}
\begin{theorem}
\label{co:7}
Given a non-constant meromorphic solution $f$ of $ a_0+a_1f+\cdots+a_pf^p=0$ such that $T(r,a_i)=S(r,f)$, then $a_0=a_1=\cdots=a_p=0.$
\end{theorem}
\begin{theorem}[See \cite{zhangliao13}]
\label{th:11}
Let $f$ and $P$ be as in Theorem \ref{th:4}. If $P$ has only one dominant term (Definition \ref{defdiffpoly}) and $f$ satisfies $P=0$, where $T(r,a_\lambda)=S(r,f)$ and $N(r,f)=S(r,f)$, then $f$ is rational.
\end{theorem}

Now we present our main lemmas concerning meromorphic solutions of (\ref{eq:1}). Below, $x,y,z$ are the meromorphic solutions of (\ref{eq:1}). We first prove some inequalities related the growth of $x,y,z$.

\begin{lemma}
    \label{lem1}
    If $x,y,z\notin\mathbb{C}$, then $m(r,u)+S(r,u)\leq m(r,v)+S(r,v),$ for $u,v\in\{x,y,z\}$. 
\end{lemma}
\begin{proof}
    Without loss of generality, $u=x,v=y$. From (\ref{eq:1,1}) and (\ref{eq:1,2}), we have $x=y'/y-Ax'/x+ACy+A\lambda-\mu$, so Theorems \ref{th:1} and \ref{th:2} imply $m(r,x)+S(r,x)\leq m(r,y)+S(r,y)$.
\end{proof}
Now we show that the proximity functions $m(r,x)$, $m(r,y)$ and $m(r,z)$ have slow-growth when $D\neq0$.
\begin{lemma}
\label{lem2}
Let $x,y,z\notin\mathbb{C}$ and $D\neq0$. Assume that at least one of $x$, $y$, $z$ has a pole.
    \begin{enumerate}[label=(\alph*)]
        \item If $(x'-\lambda x)(y'-\mu y)(z'-\nu z)\not\equiv0$, then $T(r,x)\asymp T(r,y)\asymp T(r,z)$ and $m(r,u)=S(r,v), $ for $u,v\in\{x,y,z\}$. 
        \item If $z'=\nu z$, then $T(r,x)\asymp T(r,y)$ and $m(r,u)=S(r,w)$, where $u\in\{x,y,z\}$ and $w\in\{x,y\}$.
    \end{enumerate}
\end{lemma}
\begin{proof}
Let $u\in\{x,y,z\}$. Since $D\neq0$, (\ref{eq:1}) implies $u=a_{ux}x'x^{-1}+a_{uy}y'y^{-1}+a_{uz}z'z^{-1}+a_{u}$
where $a_{ux},a_{uy},a_{uz},a_u\in\mathbb{C}$. Theorem \ref{th:2} shows that
\begin{equation}
    \label{eq:27}
    m(r,u)=S(r,x)+S(r,y)+S(r,z).
\end{equation}
Assume that $(x'-\lambda x)(y'-\mu y)(z'-\nu z)\not\equiv0$. Eliminating $z,y'$ in (\ref{eq:1}) yields 
\begin{equation}
    \label{lem2eq:2}
    {\begin{aligned}
     C(AC+1)y^2=\left\{(AC+1)f+C\left(\mu-\nu+x-Bx\right)\right\}y+(Bx+\nu)f+f',
\end{aligned}}
\end{equation}
where $f=x'/x-\lambda$.
If $C(AC+1)\neq0$, then Theorem \ref{th:2} implies $2T(r,y)\leq  T(r,y)+O(T(r,x))$ i.e. $T(r,y)=O(T(r,x))$. If $C(AC+1)=0$, then (\ref{lem2eq:2}) becomes $P(x)y=Q(x)$ where $P,Q\in \mathbb{C}[x]$. Here, $P=0$ implies $fD=0$. Hence, $y=P(x)/Q(x)$ and $T(r,y)=O(T(r,x))$. In either case, $T(r,y)=O(T(r,x))$. A similar argument yields $T(r,x)=O(T(r,y))$. Therefore, $T(r,x)\asymp T(r,y)$. Analogously, $T(r,z)\asymp T(r,y)$ and from (\ref{eq:27}), it follows that $m(r,u)=S(r,v)$. If $z'=\nu z$, then the proof for $T(r,x)\asymp T(r,y)$ still holds since both $x,y,$ must have poles by hypothesis. Since $z'=\nu z$, we have $m(r,z)=T(r,z)$. Setting $u=z$ in (\ref{eq:27}) gives $m(r,z)=S(r,w)$. Again, from (\ref{eq:27}), $m(r,u)=S(r,w)$.
\end{proof}
In Lemma \ref{lem2}, $T(r,x)\asymp T(r,y)$ implies that if a meromorphic function $f$ satisfies $T(r,f)=S(r,x)$, then $T(r,f)=S(r,y)$ and vice-versa. This will be useful in sections \ref{abc+1=0} and \ref{abc+1not0}, where we can treat a small function with respect to either $x$ or $y$.

Now we prove Lemma \ref{lem3}, which will be used to show that the constructed functions are constant. First, we require the following definition.
\begin{definition}
\label{defsingpole}
Let $S_0,S_1,S_2,S_3$ be subsets of $\mathbb{C}$ whose elements are the $\mathbf{p_0},\mathbf{p_x},\mathbf{p_y},\mathbf{p_z}$ poles, respectively. 
For $i\in\{0,1,2,3\}$, we define
\[
N_i(r)=\int_0^r\frac{\# (S_i\cap\{t:|t|\leq s\})-\# (S_i\cap\{0\})}{s}\dd s+\# (S_i\cap\{0\})\ln r.
\]
where $\#S$ denotes the cardinality of the set $S$. Since the poles of $x,y,z$ are simple, there is no need to consider multiplicities in $\# (S_i\cap\{t:|t|\leq r\})$. 
\end{definition}
\begin{lemma}
\label{lem3}
Suppose that $m(r,u)=S(r,y)$ and $u\in\{x,y,z\}$.
Let $f\in \mathbb{C}[x,y,z]$ such that $N(r,f)=S(r,y)$. If for some $i$ and $f_0$ we have $f=f_0+O(t-t_0)$ for all $t_0\in S_i$, then either $f=f_0$ or $N_i(r)=S(r,y)$. 
\end{lemma}
\begin{proof}
If $f\neq f_0$, then $f$ takes the value $f_0$ on $S_i$ i.e. $N(r,1/(f-f_0))\geq N_i(r)$. Since $N(r,f)=S(r,y)$ and $f\in \mathbb{C}[x,y,z]$, it follows that $T(r,f)=S(r,y)$. Using Theorem \ref{firstfund}, we obtain $S(r,y)=T(r,f)+O(1)=T(r,1/(f-f_0))\geq N_i(r)$.
\end{proof}

\begin{definition}
\label{deflmnp}
A function $f$ belongs to Class W, denoted by $f\in W$, if $f$ is rational, elliptic, or simply-periodic of the form $f=R(\tau)$, where $R$ is a rational function of $\tau=\e^{\delta t}$ for some $\delta\neq0$. 

For a rational function $R=P/Q$ where $P,Q$ are polynomials such that $\gcd(P,Q)=1$, we denote $\deg R=\max(\deg P,\deg Q)$.

When $f$ is simply-periodic in Class $W$, then $f(\tau)=P(\tau)/Q(\tau)$ where $P,Q$ are polynomials such that $\gcd(P,Q)=1$ and $f$ is said to take:
\begin{enumerate}[label=(\alph*)]
    \item min-form if $\deg P=0$
    \item mid-form if $0<\deg P<\deg Q$
    \item or max-form if $\deg P=\deg Q$.
\end{enumerate}

When $x,y,z\in W$, then we define $\mathfrak{D}$ to be the fundamental region. We denote $l,m,n,p$ as $\#S_0,\#S_1,\#S_2,\#S_3$ in $\mathfrak{D}$, respectively. 
\end{definition}
\begin{remark}
\label{remarkminmaxmid}
Suppose that $f=R(\e^{\delta t})$ takes min-form. By considering $f$ as a rational function in $\tilde{\tau}=\e^{-\delta t}$, then it can be seen that $f=\tilde{R}(\e^{-\delta t})$ for some rational $\tilde{R}$, takes the max-form, and vice-versa. This means that if there is a case where $f$ takes min-form, there is no need to discuss the case where $f$ takes max-form, and vice-versa.
\end{remark}
\begin{remark}
\label{remarksphere}
Recall that every non-constant rational function $R$ takes all values in $\mathbb{C}\cup\{\infty\}$ $\deg R$-times counting multiplicities on $\mathbb{C}\cup\{\infty\}$.  
\end{remark}

Another important lemma is the following.
\begin{lemma}
\label{anslem}
Suppose that $D\neq0$ and $x,y,z\in W$. Let $\chi_i=1$ if $S_i\neq\emptyset$; otherwise $\chi_i=0$. Let 
\begin{equation}
\label{anslemeq:00}
M_x=l\chi_{0}+n\chi_{2}+p\chi_{3},\quad V_x=lx_0\chi_{0}+p\hat{x}_0\chi_{3}+n\Check{x}_0\chi_{2},
\end{equation}
where $x_k$, $\hat{x}_k$, $\breve{x}_k$ and $\check{x}_k$ are the $x_k$ term in (\ref{eq:2}) of the $\mathbf{p_0},\mathbf{p_z},\mathbf{p_x},\mathbf{p_y}$ poles, respectively.
\begin{enumerate}[label=(\alph*)]
\item The functions $x,y,z$ are all elliptic, all simply-periodic, or all rational. Here simply-periodic functions are those functions that have exactly one independent period.
\item If $x$ is elliptic, then $\lambda=\mu=\nu$.
\item If $x$ is simply-periodic, then
    \begin{equation}
    \label{anslemeq:1}
    x=\sum_{j=1}^l\frac{x_0\delta \tau\chi_{0}}{\tau-\tau_{j}}+\sum_{j=1}^p\frac{\hat{x}_0\delta \tau\chi_{3}}{\tau-\hat{\tau}_{j}}+\sum_{j=1}^n\frac{\check{x}_0\delta \tau\chi_{2}}{\tau-\check{\tau}_{j}}+d_x,\quad \tau=\e^{\delta t}.
    \end{equation}
    When $x$ has no zeros, then 
    \begin{equation}
    \label{anslemeq:01}
    x=\frac{\acute{x}_0\tau^{m_x}}{\prod_{j=1}^l(\tau-\tau_{j})\prod_{j=1}^p(\tau-\hat{\tau}_{j})\prod_{j=1}^n(\tau-\check{\tau}_{j})},
    \end{equation}
    where $0\leq m_x\leq M_x=l+n+p$. Thus,
    \begin{enumerate}[label=(\roman*)]
        \item when $x$ takes min-form, then $m_x=0,d_x=-\delta V_x$. 
        \item when $x$ takes max-form, then $d_x=0,m_x=M_x$ and $\acute{x}_0=\delta V_x$.
        \item when $x$ takes mid-form, then $V_x=0,d_x=0,m_x\in(0,M_x)$.
    \end{enumerate}
    When $x$ has zeros, then
    \begin{equation}
    \label{anslemeq:001}
    x=\frac{\acute{x}_0\tau^{m_x}\prod_{j=1}^m(\tau-\breve{\tau}_{j})^\alpha}{\prod_{j=1}^l(\tau-\tau_{j})\prod_{j=1}^p(\tau-\hat{\tau}_{j})\prod_{j=1}^n(\tau-\check{\tau}_{j})},
    \end{equation}
    where $0\leq m_x+\alpha m\leq l+n+p$.
\item If $x$ is rational, then 
    \begin{equation}
    \label{anslemeq:2}
    x=\sum_{j=1}^l\frac{x_0\chi_{0}}{t-t_{j}}+\sum_{j=1}^p\frac{\hat{x}_0\chi_{3}}{t-\hat{t}_{j}}+\sum_{j=1}^n\frac{\check{x}_0\chi_{2}}{t-\check{t}_{j}}+d_x.
    \end{equation}
    \begin{enumerate}[label=(\roman*)]
        \item If $x$ has no zeros, then $d_x=0,V_x=0$.
        \item If $x,y$ have no zeros, then $d_z+\lambda=Ad_z+\mu=\nu=0$ where $d_z$ is an analogue of $d_x$ for $z$. Moreover, if $d_z\neq0$, then $z=P/Q$ with $\deg P=\deg Q$.
    \end{enumerate}
\end{enumerate}
In the above formulas, if there is no singularity of type $\mathbf{p_0}$, we set $l=0$ and remove the factor and summation containing these poles from the formulas (similarly for other poles).
\end{lemma}
When $x,y,z\in W$ and are not elliptic, the formulas are given in (\ref{anslemeq:1}) and (\ref{anslemeq:2}). We use these formulas to find all undetermined variables $\tau_{i_j}$. However, to do so, we need to know upper bounds for $l,m,n,p$, which is equivalent to knowing the number of choices that the resonance parameters can admit. This can be determined using the fact that the constructed functions are constant and by expanding the Laurent series about each type of the singularity of the constructed functions.   

The proof of Lemma \ref{anslem} relies on Lemma \ref{subanslem}, which itself is used in some cases in section \ref{abc+1not0} to conclude that $x,y,z\in W$.
\begin{lemma}
\label{subanslem}
Suppose that $D\neq0$ and there exists at least two types of singularities, say $\mathbf{p_0},~\mathbf{p_z}$ poles.
\begin{enumerate}[label=(\alph*)]
    \item If $x,y,z$ have known expansions about one of them, say $\mathbf{p_0}$, then $x,y,z$ are either all elliptic, all simply-periodic or all have unique expansions about each point $t_0\in S_0$. 
    \item If $N_0(r),N_3(r)\neq S(r,y)$, then $\#S_0,\#S_3$ are both finite or both infinite in $\mathfrak{D}$. 
    \item If $x,y,z$ are simply-periodic and the number of poles are finite in $\mathfrak{D}$, then $x,y,z\in W$. Moreover, $T(r,x)\asymp r,m(r,x)=O(1),$ and $x(\tau)=P(\tau)/Q(\tau),\tau=\e^{\delta t}$ such that $\deg P\leq \deg Q$. This similarly holds for $y,z$.
    \item If $\ln r\asymp N_0(r)\neq S(r,y)$, then $x,y,z$ are rational and have unique Laurent series about each point.
\end{enumerate}
All the above conclusions are similar if the types of singularities are changed.
\end{lemma}
\begin{proof}
If $x$ has the same Laurent series about two points, then $x,y,z$ are periodic. If there are three non-collinear such points, then $x,y,z$ are elliptic. Thus, $(a)$ is proved. If $N_0(r)=O(\ln r)$ and $\#S_3=\infty$ i.e. $\ln r=o(N_3(r))=S(r,y)$, then $N_0(r)=S(r,y)$ and so $(b)$ is proved. If $\ln r\asymp N_0(r)\neq S(r,y)$ i.e. there are finitely many $\mathbf{p_0}$ poles, then $(b)$ shows that there are finitely many poles. The uniqueness of the Laurent expansion follows from the proof of $(a)$. Then Theorem \ref{th:rational} and Lemma \ref{lem2} show that $x,y,z$ are rational, proving $(c)$. Now we prove $(d)$. Lemma \ref{lem2} yields $m(r,x)=S(r,x)$. Since there are finitely many poles, then $N(r,x)=O(r)$ and so $m(r,x)=O(r)$. By Theorem \ref{firstfund}, we have $N(r,1/(x-a))=O(r)$ for all $a\in\mathbb{C}$. Arguments in \cite[Th.3 first case]{eremenko1986meromorphic} show that $x\in W$ and $T(r,x)\asymp r$. Here if $\deg P>\deg Q$, then $m(r,x)\asymp r$; otherwise $m(r,x)=O(1)$ \cite[Ch.1]{hayman1964meromorphic}. As $m(r,x)=S(r,x)$, then $m(r,x)=O(1)$ and $\deg P\leq\deg Q$.
\end{proof}
\begin{proof}[Proof of Lemma \ref{anslem}]
Here $(a)$ follows from Lemma \ref{subanslem}$(a)$. If $x$ is elliptic, then $x$ has zeros, which shows that the $\mathbf{p_x}$ poles exist (Remark \ref{zerosxyz}), and $\mu=\nu$ (analogue of (\ref{eq:13.1})). Similarly for $y,z$, so $\lambda=\mu=\nu$, proving $(b)$. Part $(c)$ follows from a straightforward substitution and Lemma \ref{subanslem}$(c)$. Similarly, $(\ref{anslemeq:2})$ and $(d)(i)$ follow from a straightforward substitution and Lemma \ref{subanslem}$(d)$. To see $(d)(ii)$, we use (\ref{eq:1}) and substitute the expressions for $x,y,z$ in (\ref{anslemeq:2}) and its analogues into $Cy+z+\lambda,Az+x+\mu,Bx+y+\nu$ to obtain
\begin{equation}
\label{proanslemeq:1}
{\begin{aligned}
\frac{x'}{x}=\sum_{j}\frac{-1}{t-t_{j}}+d_z+\lambda,~\frac{y'}{y}=\sum_{j}\frac{-1}{t-t_{j}}+Ad_z+\mu,~\frac{z'}{z}=\sum_{j}\frac{k_{j}}{t-t_{j}}+\nu,    
\end{aligned}}    
\end{equation}
where $k_{j}\in\mathbb{Z}$. In (\ref{proanslemeq:1}), there are no $d_x,d_y$ terms since $d_x=d_y=0$ follows from $(d)(i)$. Integrating (\ref{proanslemeq:1}) shows that $d_z+\lambda=Ad_z+\mu=\nu=0$. If $d_z\neq0$, then the analogue of (\ref{anslemeq:2}) for $z$ shows that $\deg P=\deg Q$, proving $(d)(ii)$.
\end{proof}
Lemma \ref{subanslem}$(b)$ states that if there are only finitely many Laurent series of $x,y,z$ corresponding to one type of pole, then $x,y,z$ have only finitely many Laurent series overall (Eremenko's finiteness property). In the proof of Lemma \ref{subanslem}$(a),(c),(d)$, we showed that the slow-growth condition (Lemma \ref{lem2}) and Eremenko's finiteness property of $x,y,z$ imply that $x,y,z\in W$ of the same type. For a comparison with \cite{eremenko2005,halburdwang14}, see section \ref{discussion}.

\section{Holomorphic solutions}
\label{holosolu}
If $x,y,z$ are all constant, then we set $x'=y'=z'=0$ in (\ref{eq:1}) and solve for $x,y,z$. Now we assume that at least one of $x,y,z$ is non-constant. Here, $T(r,u)=m(r,u)$ for all $u\in\{x,y,z\}$.
If $x=y=0$, then (\ref{eq:1,1}) yields $z=K\e^{\nu t},$ where $K\in\mathbb{C}$.
If $yz\not\equiv0,x=0$, then (\ref{eq:1}) becomes $y'/y=Az+\mu,z'/z=y+\nu$ which yields holomorphic solutions when $A=0$.
If $A\neq0$, then $Az=y'/y-\mu,~y=z'/z-\nu$, which, by Theorems \ref{th:1} and \ref{th:2} imply $T(r,z)=S(r,y), T(r,y)=S(r,z)$.
Thus, $T(r,y)=S(r,y),T(r,z)=S(r,z)$ i.e. $y,z$ are constant. So, we assume that $xyz\neq0$. Remark \ref{zerosxyz} asserts that $x,y,z$ have no zeros. If $C=0$, then (\ref{eq:1,1}) implies $T(r,z)=S(r,x)$, and (\ref{eq:1,2}) yields $T(r,x)=S(r,y)$. Thus, $T(r,z)=S(r,y)$. In (\ref{eq:1,3}), we obtain $T(r,y)=S(r,z)=S(r,y)$ i.e. $y\in\mathbb{C}$. Then, (\ref{eq:1,1}) and (\ref{eq:1,2}) yield $Ax'=x(A\lambda-x-\mu)$, which implies that $x$ has poles, unless $x$ is constant. Hence, $z$ is also constant. Similarly when $AB=0$, we only obtain $x,y,z\in\mathbb{C}$. Thus, we may assume that $ABC\neq0$. Let $f=x'/x,g=y'/y$, then $T(r,f)=S(r,x),T(r,g)=S(r,y)$. Lemma \ref{lem1} shows that $T(r,z)\asymp T(r,x)\asymp T(r,y)$ and so $T(r,f)=S(r,y)$. Expressing $f$, $f'$ and $g$ in terms of $x$, $y$ and $z$ and eliminating $x$ and $z$ yields $CDy^2+a_1y+a_0=0$ where $a_1=(2D-1)(f-\lambda)+C\{(1-B)g+B\mu-\nu\}$ and $a_0=f'+(f-\lambda)\{AB(f-\lambda)-Bg+B\mu-\nu\}$. It follows that $D=a_1=a_0=0$, which implies that
\begin{equation}
\label{holoeq2}
g'(AC+1)-g^2(AC+1)+g(\mu+AC\nu)=0.
\end{equation}
If $g=0$, then $y\in\mathbb{C}$. Here, $a_1=0$ implies $f\in\mathbb{C}$. So $z,x\in\mathbb{C}$.
If $AC+1=0$, then $B=1$. Simplifying (\ref{holoeq2}) and $a_1=0$ yields $\mu=\nu$ and $f=\lambda$. Then $x=K\e^{\lambda t},z=-Cy,K\in\mathbb{C}$ and (\ref{eq:1,2}) implies $y'/y=y+K\e^{\lambda t}+\mu$. Let $h:=y+K\e^{\lambda t}$, then $T(r,h)=m(r,y'/y)=S(r,y)$. Moreover, $y(h-\lambda+\mu)+\lambda h-h'=0$. So, Theorem \ref{co:7} yields $\lambda=\mu,h=0$. Thus we obtain the solution $y=-K\e^{\lambda t},x=K\e^{\lambda t},z=CK\e^{\lambda t}$. Finally, if $g(AC+1)\not\equiv0$, then (\ref{holoeq2}) shows that $g$ is entire only if $g=(AC\nu+\mu)(AC+1)^{-1}$.
Then, $a_1=0$ reduces to $f=\lambda$. So $x=K\e^{\lambda t},z=-Cy$ where $K\in\mathbb{C}$. Since $g=y'/y$, then $y=L\e^{gt},z=-CL\e^{gt}$ where $L\in\mathbb{C}$. Substituting $x,y,z$ into $Az+x=g-\mu$ yields $g=\lambda=\mu=\nu,L=-BK$. Therefore, we obtain the solution: $y=-BK\e^{\lambda t}$, $x=K\e^{\lambda t}$, $z=BCK\e^{\lambda t}$ with $\lambda=\mu=\nu$.

\section{Case $D=0$}
\label{abc+1=0}
Recall that when $D=0$, if there is at least one $\mathbf{p_0}$ pole, then there can be no other type of pole. Each subsection heading below specifies the exact types of poles that must be present in the case considered, with no other poles permitted.

\subsection{$\mathbf{p_0}$ poles}
\label{0-3pole}
The resonance condition (\ref{eq:8}) implies that $C\ne 0$, $C\ne 1$, $B=(C-1)/C$ and $A=1/(1-C)$. Using $N(r,x)=N(r,y)=N(r,z)$ and Lemma \ref{lem1}, we obtain $T(r,x)\asymp T(r,y) \asymp T(r,z)$. Let 
\begin{equation}
\label{0-x-y=0eq:1}
f=x'x^{-1}-y'y^{-1}=-x+Cy+(1-A)z+\lambda-\mu.
\end{equation}
Consequently, $T(r,f)=m(r,f)=S(r,y)$. Assume that $\lambda,\mu$ and $\nu$ are pairwise distinct.
Using equation (\ref{eq:1}), we can express $f'$ and $f''$ in terms of $x$, $y$ and $z$. Eliminating $x$ and $z$ from these expressions and the expression (\ref{0-x-y=0eq:1}) for $f$, we obtain an equation of the form $Ca_1y+a_0=0$, where $a_1=f'b_0+b_1f+b_0$ and $a_0\in\mathbb{C}[C,f,f',f'',\lambda,\mu,\nu]$. Here, $b_0=(C-1) \mu -C \nu +\lambda $ and $b_1,b_2\in\mathbb{C}[C,\lambda,\mu,\nu]$.
It follows that $a_0=a_1=0$. If $b_0=0$, then $a_0=a_1=0$ reduces to $f=0$. Since $f'-\lambda f=(Cy+z+\lambda)(\mu-\lambda)=x'(\mu-\lambda)/x$, then $f=0$ implies $x'=0$; a contradiction. Hence, $b_0\neq0$, and we may express $f',f''$ in terms of $f$ from $a_1=0$.
Eliminating $f',f''$ from $a_1=a_0=0$, we can determine $f$ and the relations between $\lambda,\mu,\nu$. Using the fact that $f,f'$ are polynomials in $x,y,z$, we express $x,y$ in terms of $z$ and substitute these into (\ref{eq:1}) to verify whether (\ref{eq:1}) holds. Following this procedure, we find the following solutions: $b_0=BC\mu-C\nu+\lambda\neq0$ and either
\begin{subequations}
\begin{align}
\label{0-3poleeq:01}
&x=B^{-1}\nu\lambda^{-1} z,~y=A\lambda^{-1}(z+\lambda)(\lambda-\nu),~z'=A\lambda^{-1}(\lambda-C\nu)z(z+\lambda),~\mu=0,\\   
&\ \mbox{or}\nonumber
\\
\label{0-3poleeq:02}
&x=C\mu^{-1}(\nu-\mu)(Az+\mu),~y=A\nu\mu^{-1} z,~z'=C\mu^{-1}(\nu-B\mu)z(Az+\mu),~\lambda=0.
\end{align}
\end{subequations} 
If $\lambda=\mu\neq\nu$, then $f'-\lambda f=AC(\lambda-\nu)z$. This implies $T(r,z)=S(r,z)$ i.e. $z$ is constant, which is a contradiction. Using (\ref{transmaps}), it is enough just to consider the case $\lambda=\mu=\nu$. We may assume that $\lambda=\mu=\nu=0$ by using the following transformation \cite{bountis83},\cite[p.272]{bountis84}  
\begin{equation}
\label{eq:22}
    T=\e^{\lambda t},~
    X(T)=\lambda^{-1}x(t)\e^{-\lambda t},~ Y(T)=\lambda^{-1}y(t)\e^{-\lambda t},~ Z(T)=\lambda^{-1}z(t)\e^{-\lambda t}.
\end{equation}
Let $\dot{}$ be the derivative with respect to $T$, then (\ref{eq:1}) becomes
\begin{subequations}
    \label{eq:23}
    {\begin{align}
    \label{eq:23.1}
    &\dot{X}=X(CY+Z),\\
    \label{eq:23.2}
    &\dot{Y}=Y(AZ+X),\\
    \label{eq:23.3}
    &\dot{Z}=Z(BX+Y).
\end{align}}
\end{subequations}
The transformation (\ref{eq:22}) preserves the types of singularities of $x,y,z$ i.e. $X,Y,Z$ have the same types of poles as $x,y,z$ with the exception at $T=0$.
However, $X,Y,Z$ may not be meromorphic at $T=0$.
Let $F=\dot{X}X^{-1}-\dot{Y}Y^{-1},G=\dot{Y}Y^{-1}-\dot{Z}Z^{-1}$, then (\ref{eq:23}) implies $\dot{F}=\dot{G}=0$ and $F=-CG$. Hence, $F,G\in\mathbb{C}$, and $Y=H_2Z\e^{GT},X=H_1Y\e^{FT}=H_1H_2\e^{(1-C)GT}Z$ for some $H_1,H_2\in\mathbb{C}$. Substituting these into (\ref{eq:23}), we obtain the solution:
\begin{equation}
\label{0-3poleeq:1}
\begin{aligned}
&X=H_1H_2\e^{(1-C)GT}Z,~Y=H_2\e^{Gt}Z,~\dot{Z}=Z^2(BH_1H_2\e^{(1-C)GT}+H_2\e^{GT}),\\&G=Z(A-ABH_1H_2\e^{(1-C)GT}-H_2\e^{GT}).    
\end{aligned}
\end{equation}
As $X,Y,Z$ in (\ref{0-3poleeq:1}) are meromorphic, inverting (\ref{eq:22}) yields $x,y,z$ are also meromorphic.

\subsection{$\mathbf{p_z}$ poles}
\label{0-x-y=0}
We have the resonance conditions $B=-1/C-\gamma,\lambda=\mu$.
Since $N(r,x)=N(r,y)$, Lemma \ref{lem1} shows that $T(r,x)\asymp T(r,y)$. Here, the function $f$ defined in (\ref{0-x-y=0eq:1}) remains entire.

If $\gamma=0$, then $A=1$ and $z$ has no zeros. Here, $f'=(z+\lambda)f$. If $f=0$, then $x=Cy$ and $z'=z\nu$ i.e. $z=L\e^{\nu t}$ where $L\in\mathbb{C}$. In terms of $\xi=1/x$, (\ref{eq:1}) becomes $\xi'+(L\e^{\nu t}+\lambda)\xi+1=0$. If $f\neq0$, then $z=f'/f-\lambda$, and $T(r,z)=m(r,z)=S(r,f)$.
From (\ref{eq:1,3}), we obtain $f=Cz'/z-C\nu$ which implies $ T(r,f)=S(r,z)=S(r,f)$ i.e. $f$ is constant and $f'=0$. Therefore, $z=-\lambda$ and $f=-C\nu$. Thus, $x=K\e^{-C\nu t}y$, where $K\in\mathbb{C}$. Using (\ref{0-x-y=0eq:1}), we obtain the solution $y=C\nu/(K\e^{-C\nu t}-C),x=K\e^{-C\nu t}y,z=-\lambda$. 

If $\gamma\geq1$, then $A\neq1$. Assume that $z\not\equiv0$. We eliminate $y,x'$ in (\ref{eq:1}) to obtain
\begin{equation*}
    \label{abc1x-y>0eq02}
    (1-A)x^2=(A-1)x\left(z-C\nu+Cz'z^{-1}\right)+(\lambda+Az)\left(z'z^{-1}-\nu\right)+\left(z'z^{-1}\right)'.
\end{equation*}
Arguing similarly as in Lemma \ref{lem2}, we conclude that $T(r,x)= O(T(r,z)),T(r,y)= O(T(r,z))$. Expressing $f$, $f'$ and $f''$ in terms of $x$, $y$ and $z$ and eliminating $x$ and $y$ yields $(A-1)A^2dC^3z(a_3z^3+a_2z^2+a_1z+a_0)=0$ where $d=AC-A+1,~a_3=(A-1)d^2f$. Here, $d\neq0$ following Remark \ref{coexistsabc1=0}. It follows that that $a_3=0$. So $f=0$. Using (\ref{0-x-y=0eq:1}), we obtain $x=Ky$ for some $K\in\mathbb{C}$, and $(1-A)z=x-Cy=(K-C)y$. Since $N(r,z)=S(r,y)$, then we must have $A=1,K=C$; a contradiction. Hence, $z=0$ and (\ref{eq:1}) reduces to the 2 dimensional case. Here, $f'=\lambda f$. Using (\ref{eq:1}) and (\ref{0-x-y=0eq:1}), we obtain the solution $x=1/\xi$, where $\xi$ solves $\xi'+(K\e^{\lambda t}+\lambda)\xi+1=0$ and $y=(K\e^{\lambda t}+x)/C$, where $K\in\mathbb{C}$.

\subsection{$\mathbf{p_z},\mathbf{p_x}$ poles}
\label{0-x-y-y-z==0}
The resonance conditions are $\lambda=\mu=\nu, -B-1/C=\gamma,$ $-C-1/A=\alpha$. Let
\begin{equation}
\label{firstintegrals}
{\begin{aligned}
f=-x+Cy-ACz\mbox{\ \ and\ \ } h=x^Ay^{-1}z^{-AC}.    
\end{aligned}}
\end{equation}
In \cite{bountis84}, it was shown that 
\begin{equation}
\label{firstintegrals1}
f=K\e^{\lambda t}\mbox{\ \ and\ \ }h=L\e^{(A-1-AC)\lambda t},
\end{equation}
where $K,L\in\mathbb{C}$. When $\lambda=0$, both $f$ and $h$ are first integrals.

If $\alpha=0$ and $\gamma\geq0$, then $B=1$ and $C=-1/A,A=\gamma+1$. Using (\ref{firstintegrals}),(\ref{firstintegrals1}) and (\ref{eq:1}), we obtain the solution $x'=x(x+K\e^{\lambda t}+\lambda),z=(K\e^{\lambda t}+x)(CL\e^{-A\lambda t}x^{A}+1)^{-1},y=L\e^{-A\lambda t}x^{A}z$. By using (\ref{swapmaps3}), it is unnecessary to discuss the case where $\gamma=0,\alpha\geq0$.

If $\alpha\gamma>0$, then we use $f$ to eliminate $y$ in (\ref{eq:1,1}) and (\ref{eq:1,3}). Thus, we obtain
\begin{equation}
\label{eq:25''}
{\begin{aligned}     
x'=x\left(x+\frac{\alpha(\gamma+1)}{\alpha\gamma-1}z+f+\lambda\right)
\mbox{\ \ and\ \ }
z'=z\left(\frac{(\gamma+1)z}{1-\alpha\gamma}-\gamma x-\frac{(\gamma+1)f}{\alpha+1}+\lambda\right).
\end{aligned}}
\end{equation}
Differentiating (\ref{eq:25''}), we obtain the following second-order ODEs:
\begin{subequations}
    \label{eq:25'}
    {\begin{align}
    \label{eq:25'.1}
      x''=&\frac{\alpha-1}{\alpha}\frac{(x')^2}{x}+\frac{\alpha\gamma-1}{\alpha}x^3+c_1xx'+c_2x'+c_3x^2+c_4x,\\
    \label{eq:25'.2}  
      z''=&\frac{\gamma-1}{\gamma}\frac{(z')^2}{z}+\frac{(\gamma+1)^2}{\gamma(\alpha\gamma-1)}z^3+c_5zz'+c_6z'+c_7z^2+c_8z,
    \end{align}}
\end{subequations}
where the coefficients $c_i$ depend only on $f,\lambda,\alpha,\gamma$. 

Now we reduce to the case $\lambda=0$ by using (\ref{eq:22}).
Meromorphic solutions of (\ref{eq:1}) correspond to solutions of (\ref{eq:23}) which are meromorphic on the covering surface of $\mathbb{C}-\{0\}$. Using (\ref{firstintegrals1}), we obtain $X-CY+ACZ=K$ and $X^{\gamma+1}Y^{(\alpha-1)(\gamma-1)}Z^{\alpha+1}=L^{\alpha+1}$. Using these first integrals, we eliminate $X,Z$ in (\ref{eq:23.2}) to obtain $P(Y,\dot{Y})=0$, where $P\in\mathbb{C}[Y,\dot{Y}]$. It was shown in \cite{painleve1902} that the solutions of $P(Y,\dot{Y})=0$ will at most have algebraic branch points. So $Y$ is meromorphic on the covering surface of $\mathbb{C}-\{0\}$ with a possible algebraic branch at $T=0$. So are $X,Z$. Hence, there exists $k\in\mathbb{N}$ such that $X(T^k),Y(T^k),Z(T^k)$ are meromorphic. The task is to find meromorphic functions $\tilde{x}=x\circ t^k,\tilde{y}=y\circ t^k,\tilde{z}=z\circ t^k$ where $x,y,z$ are solutions of (\ref{eq:1}) when $\lambda=0$. Here, $f$ is constant and $h^{\alpha+1}=\Tilde{x}^{\gamma+1}\Tilde{y}^{(\alpha-1)(\gamma-1)}\Tilde{z}^{\alpha+1}=L^{\alpha+1}$. In (\ref{eq:25'.1}), $x,x',x''$ are replaced by $\Tilde{x},\Tilde{x}'k^{-1}t^{1-k},\{\Tilde{x}''-(k-1)\Tilde{x}'t^{-1}\}k^{-2}t^{2-2k}$ respectively. The coefficient of $x^3$ in (\ref{eq:25'.1}) vanish iff $\alpha=\gamma=1$ which is equivalent to $ABC=1$. So this coefficient does not vanish and since $m(r,1/t)=O(1)$, we can apply Theorem \ref{th:4} to (\ref{eq:25'}) to obtain $m(r,\tilde{x})=S(r,\tilde{x})$. Arguing similarly, we obtain $m(r,\tilde{z})=S(r,\tilde{z})$. If $\tilde{x}$ is rational, then (\ref{eq:25''}) shows that $\tilde{z}$ is also rational. Using $f$ in (\ref{firstintegrals}) yields that $\Tilde{y}$ is rational as well. Similarly, when $\tilde{z}$ is rational. If $\Tilde{y}$ is rational but $\tilde{x},\tilde{z}$ are transcendental, then $T(r,\Tilde{y})=S(r,\Tilde{z})$ which implies $N(r,\Tilde{y})=S(r,\Tilde{z})$ and thus $T(r,\Tilde{z})\leq N(r,\Tilde{y})+m(r,\Tilde{z})=S(r,\Tilde{z})$ which is impossible. If $\Tilde{x},\Tilde{y},\Tilde{z}$ are rational, then $T(r,\Tilde{x})\asymp T(r,\Tilde{y})\asymp T(r,\Tilde{z})$. Using $C\Tilde{y}=f+\Tilde{x}+AC\Tilde{z}$ implies $m(r,\Tilde{y})=S(r,\Tilde{y})$. Moreover,
\begin{equation}
\label{0-x-y-y-z>>0eq:1}
\begin{aligned}
\Tilde{x}=\frac{\acute{x}_0t^{m_x}\prod_j(t-\breve{t}_{j})^\alpha}{\prod_{j}(t-\hat{t}_{j})},~\Tilde{y}=\frac{\acute{y}_0t^{m_y}}{\prod_j(t-\hat{t}_{j})\prod_j(t-\breve{t}_{j})},~\Tilde{z}=\frac{\acute{z}_0t^{m_z}\prod_j(t-\hat{t}_{j})^\gamma}{\prod_j(t-\breve{t}_{j})},    
\end{aligned}
\end{equation}
where $m_x,m_y,m_z\in\mathbb{Z}$. Since $m(r,u)=S(r,u)$ for $u\in\{\Tilde{x},\Tilde{y},\Tilde{z}\}$, then $m_x+m\alpha-p\leq0,m_y-m-p\leq0,m_z-m+p\gamma\leq0$. This implies $Q+(m+p)(\alpha+\gamma-2)\leq0$, where $Q:=(\gamma+1)m_x+(\alpha-1)(\gamma-1)m_y+(\alpha+1)m_z$, and $m,p$ are defined as in Definition \ref{deflmnp}.
Substituting (\ref{0-x-y-y-z>>0eq:1}) into $\Tilde{x}^{\gamma+1}\Tilde{y}^{(\alpha-1)(\gamma-1)}\Tilde{z}^{\alpha+1}=L^{\alpha+1}$, the factor $t$ cancels out i.e. $Q=0$. Hence, $(m+p)(\alpha+\gamma-2)\leq0$, which implies $\alpha=\gamma=1$. So $\Tilde{x},\Tilde{y},\Tilde{z}$ are transcendental and $\ln r=S(r,\tilde{y})$. There could possibly be poles of $\Tilde{x},\Tilde{y},\Tilde{z}$ at $0$. Following the proof of Lemma \ref{lem1} and using $m(r,1/t)=O(1)$, we can show that $m(r,u)+S(r,u)\leq m(r,\Tilde{y})+S(r,\Tilde{y})$ for $u\in\{\Tilde{x},\Tilde{z}\}$. Since $N(r,\Tilde{u})\leq N(r,\Tilde{y})+O(\ln r)=N(r,\Tilde{y})+S(r,\Tilde{y})$, then $T(r,\Tilde{u})=O(T(r,\Tilde{y}))$ for $u\in\{\Tilde{x},\Tilde{z}\}$. Again, from $C\Tilde{y}=f+\Tilde{x}+AC\Tilde{z}$, we have $m(r,u)=S(r,\Tilde{y})$ for $u\in\{\Tilde{x},\Tilde{y},\Tilde{z}\}$. Let $g=\Tilde{x}\Tilde{z}$ which has poles at most at $t=0$ and thus $T(r,g)=S(r,\Tilde{y})$. If either $\alpha>1$ or $\gamma>1$, then $g$ has zeros at those singularities and Lemma \ref{lem3} shows that $g=0$ or $N_3(r)=S(r,\Tilde{y})$ or $N_1(r)=S(r,\Tilde{y})$. However, these imply either $x=0$ or $z=0$ or a reduction of the discussion to section \ref{0-x-y=0}. So $\alpha=\gamma=1$, which is again false.

\subsection{$\mathbf{p_z},\mathbf{p_x},\mathbf{p_y}$ poles}
\label{0-2pole} 
This case is integrable \cite{bountis83,bountis84,lvsystem90,maier2013integration,goriely} and the general solution is meromorphic. Using (\ref{eq:22}), we may assume that $\lambda=0$. It was shown in \cite{bountis84,lvsystem90} that $D=0$ implies $\frac{1}{\alpha+1}+\frac{1}{\beta+1}+\frac{1}{\gamma+1}=1$. Using (\ref{conjugatemap}), we can assume that $\alpha\leq\gamma\leq\beta$ leading to $(\alpha,\beta,\gamma)=(1,5,2),(1,3,3),(2,2,2)$. 
Here, $x,z,y$ can be found by solving (\ref{eq:25''}),(\ref{eq:25'}) and (\ref{eq:1}). If $(\alpha,\gamma)=(1,2)$, then (\ref{eq:25'.1}) is equivalent to $v''=v^3-vv'-12v$ \cite[p.334, $X$ form]{ince56}. If $(\alpha,\gamma)=(1,3)$, then (\ref{eq:25'.1}) can be transformed to $v''=2v^3+f^2v/2$ which can be solved in terms of Jacobi elliptic functions \cite{armitage2006elliptic}. If $(\alpha,\gamma)=(2,2)$, then (\ref{eq:25'.2}) is $z''=(z')^2/(2z)+3z^3/2+2fz^2+f^2z/2$ \cite[p.339, $XXX$ form]{ince56}. In all cases, $x,y,z$ are meromorphic, and so by inverting (\ref{eq:22}), $x,y,z$ remain meromorphic when $\lambda\neq0$.

\section{Case $D\neq0$}
\label{abc+1not0}

\subsection{$\mathbf{p_0}$ poles}
\label{neq03*}
Let $f=x_0y-y_0x$ and $g=y_0z-z_0y$, which are entire. Lemmas \ref{lem1} and \ref{lem2} show that $T(r,\Tilde{v})=S(r,y),$ for $\Tilde{v}\in\{f,g\}$. Using $f,g$, we express $x,z$ in terms of $y$ and substitute into (\ref{eq:1}). We then eliminate $x,z,y'$ giving $a_1y+a_2=0$ and $a_3y+a_4=0$ where $a_i\in\mathbb{C}[f,g,f',g']$. So $a_i=0$, which gives
\begin{subequations}
\label{eq:28}
\begin{align}
\label{eq:28.1}
&(Az_0+1)g+(B-1)z_0f=z_0y_0(\nu-\mu),~(A-1)x_0g-(x_0+1)f=y_0x_0(\lambda-\mu),\\
\label{eq:28.2}
&y_0f'-f(g+y_0\lambda)=0,~y_0g'+g(Bf-y_0\nu)=0.
\end{align}
\end{subequations}
In (\ref{eq:28.1}), since $(Az_0+1)(x_0+1)+(A-1)(B-1)z_0x_0=Dx_0y_0z_0\neq0$, then $f,g\in\mathbb{C}$. So (\ref{eq:28.2}) implies $f(g+y_0\lambda)=g(Bf-y_0\nu)=0$. If $f=g=0$, then (\ref{eq:28.1}) and (\ref{eq:1,2}) yield
\begin{equation}
\label{neq03*eq:1}
\lambda=\mu=\nu,\quad y'= -y(y_0^{-1}y-\lambda),\quad x_0^{-1}x=y_0^{-1}y=z_0^{-1}z.
\end{equation}
Arguing similarly for other cases where $f(g+y_0\lambda)=g(Bf-y_0\nu)=0$, we obtain 
\begin{subequations}
\label{neq03*eq:2}
\begin{align}
\label{neq03*eq:2.1}
& \nu=0,x_0\lambda=Cy_0\mu, & &y'=-y(y_0^{-1}y+x_0^{-1}\mu), & &x_0^{-1}x=y_0^{-1}y=z_0^{-1}z-\mu x_0^{-1},\\
\label{neq03*eq:2.2}
& \lambda=0,y_0\mu=Az_0\nu, & &y'=-y_0^{-1}y(y+\nu), & &z_0^{-1}z=y_0^{-1}y=x_0^{-1}x-\nu y_0^{-1},\\
\label{neq03*eq:2.3}
& \mu=0,z_0\nu=Bx_0\lambda, & &y'=-y(y_0^{-1}y-z_0^{-1}\lambda), & &z_0^{-1}z=x_0^{-1}x=y_0^{-1}y-\lambda z_0^{-1}.
\end{align}
\end{subequations}

\subsection{$\mathbf{p_z}$ poles}
\label{neq0x-y}
In this case, $x,y$ have no zeros. If $\lambda=\mu$, then $z$ has zeros at every point $t_0\in S_3$. Thus, Lemma \ref{lem2} implies that $T(r,z)=S(r,y)$. Lemma \ref{lem3} shows that $N_3(r)=S(r,y)$ or $z=0$, reducing the problem to sections \ref{holosolu} and \ref{0-x-y=0} respectively. If $\lambda\neq\mu$, then $B=-1/C$ and $z$ takes value $\hat{z}_0=(\mu-\lambda)/(1-A)$ at every point $t_0\in S_3$. Lemma \ref{lem3} again yields $z=\hat{z}_0$. Then (\ref{eq:1,3}) implies $x=C(y+\nu)$. Substituting these $x$ and $z$ into (\ref{eq:1}), we obtain the following solution: $y'=y(A\hat{z}_0+Cy+C\nu+\mu),~x=C(y+\nu),~z=\hat{z}_0,\nu(\mu-A\lambda)=0$.

\subsection{$\mathbf{p_z}$ and $\mathbf{p_0}$ poles}
\label{neq0x-y3}
The resonance condition is $B=-1/C-\gamma,\gamma\in\mathbb{Z}_0^+$. Here, $x,y$ have no zeros. Let
\begin{subequations}
\label{eq:30}
\begin{align}
\label{eq:30.1}
&f=x'x^{-1}-y'y^{-1}=-x+Cy+(1-A)z+\lambda-\mu,\\
\label{eq:30.2}
&g=z'z_0+z^2-2z_1z=z_0z(Bx+y+\nu)+z^2-2z_1z,
\end{align}    
\end{subequations}
which are both entire.
About $t_0\in S_0$, we obtain $f=f_0+O(t-t_0)$ where $f_0=-x_1+Cy_1+(1-A)z_1+\lambda-\mu=2(1-A)z_1-C\nu+\lambda-\mu$.
Imposing $N_0(r)=S(r,y)$ in section \ref{neq0x-y} will not alter its proof. Thus, Lemma \ref{lem3} implies $f=f_0$.

If $\lambda\neq\mu$, then (\ref{eq:13.1}) and (\ref{eq:13.2}) imply that $\gamma=0,A\neq1$. For the $\mathbf{p_z}$ poles, we express $\hat{x}_k,\hat{y}_k,\hat{z}_k,$ for $k\geq0$ in terms of $\hat{x}_1$. Substituting (\ref{3x-y3eq:3}) into (\ref{eq:30.1}), we obtain an expression from the constant term, which can be written as $2(A-1)\hat{x}_1=f_0(1-A)-A\lambda+\mu$.
Thus, $\hat{x}_k,\hat{y}_k,\hat{z}_k,k\geq0$ are known. So $x,y,z$, and $g$ have unique expansions about $t_0\in S_3$. Then Lemma \ref{lem3} shows that $g\in\mathbb{C}$, and $z\in W$. Solutions of (\ref{eq:30.2}) have zeros unless $g=-z_1^2$ or $0$. If $g=-z_1^2$, then $z$ is a rational function, and Lemma \ref{anslem}$(d)$ shows that $\lambda=\mu=\nu=0$. Thus, $g=0$, and (\ref{eq:30.2}) which yields 
\[
z=2z_1\e^{2z_1t/z_0}(\e^{2z_1t/z_0}-\e_0)^{-1},\quad \e_0\in\mathbb{C}.
\]
So, $z,x,y$ are simply-periodic in Class $W$. Using (\ref{eq:30.1}), we obtain $x=L\e^{ft}y$ where $L\in\mathbb{C}$, and $ y=\{f-\lambda+\mu-(1-A)z\}/(C-L\e^{ft})=\{(2z_1-z)(1-A)-C\nu\}/(C-L\e^{ft})$.
This shows that $y$ has zeros unless $\nu=0$ or $f=\lambda-\mu$. Therefore, 
\begin{subequations}
\label{3x-y3eq:06}
{\begin{align}
\label{3x-y3eq:6}
    &y={-2z_1(1-A)\e_0}\{(C-L\e^{ft})(\e^{2z_1t/z_0}-\e_0)\}^{-1},&\text{if}~\nu=0,\\
\label{3x-y3eq:7}
    &y={-2z_1(1-A)\e^{2z_1t/z_0}}\{(C-L\e^{ft})(\e^{2z_1t/z_0}-\e_0)\}^{-1},&\text{if}~f=\lambda-\mu.    
\end{align}}    
\end{subequations}
We now show that $f=\pm2z_1/z_0$. If (\ref{eq:7}) has no roots in $\mathbb{N}$, then $x,y,z$ have unique expansions about any $t_0\in S_0$. Assuming the contrary, let $N\in\mathbb{N}$ solve (\ref{eq:7}), and we express $x_N,y_N$ in terms of $z_N$. Substituting (\ref{3x-y3eq:1}) into (\ref{eq:30.2}) yields an expression obtained from the coefficient of $(t-t_0)^{N-2}$, giving $z_0z_N=P$ where $P$ is a polynomial in $x_i,y_i,z_i,$ for $i<N$ and each $x_i,y_i,z_i$ are known. So $z_N=P/z_0$ and $z,x,y$ have a unique expansion about any points in $S_0$. Then $l=p=1$, and (\ref{3x-y3eq:06}) shows that $f=\pm2z_1/z_0$. 
Consider $x,y,z$ as rational functions in $\tau$, then (\ref{eq:1}) implies
\begin{subequations}
    \label{3x-yy-z3eq:7}
    {\begin{align}
    \label{3x-yy-z3eq:7.1}
    \delta\tau\dv{x}{\tau}=x(Cy+z+\lambda),\\
    \label{3x-yy-z3eq:7.2}
    \delta\tau\dv{y}{\tau}=y(Az+x+\mu),\\
    \label{3x-yy-z3eq:7.3}
    \delta\tau\dv{z}{\tau}=z(Bx+y+\nu).
    \end{align}}
\end{subequations} 
If $f=2z_1/z_0,\nu=0$, then using (\ref{3x-y3eq:6}) and $l=1$ we have $\delta=f={2z_1}{z_0^{-1}}$ and 
\begin{equation}
\label{3x-y3eq:20}
y=\frac{-2z_1(1-A)e_0}{(C-L\e^{\delta t})(\e^{\delta t}-e_0)},x=\frac{-2z_1(1-A)e_0L\e^{\delta t}}{(C-L\e^{\delta t})(\e^{\delta t}-e_0)},z=\frac{2z_1\e^{\delta t}}{\e^{\delta t}-e_0}
\end{equation}
where $x$ takes mid-form. By Lemma \ref{anslem}$(c)$, we obtain $V_x=0$ which yields $A=2/(2-C)$. Substituting (\ref{3x-y3eq:20}) into (\ref{3x-yy-z3eq:7.1}),(\ref{3x-yy-z3eq:7.2}), we obtain $\mu=0$. In this case, we obtain the solution: $\mu=\nu=0,\delta=\lambda(1-C)^{-1},A=2(2-C)^{-1}$ and
\begin{equation}
\label{3x-y3eq:22}
\begin{aligned}
x=\frac{ACz(Az+2\delta)}{2\lambda-ACz},y=\frac{\lambda(Az+2\delta)^2}{2\delta(2\lambda-ACz)},z=\frac{-2\delta\e^{\delta t}}{A\lambda(\e^{\delta t}-e_0)}.    
\end{aligned}
\end{equation}
Arguing similarly in the case $f=-2z_1/z_0,\nu=0$ yields the solution: $\lambda=\nu=0$ and
\begin{equation}
\label{3x-y3eq:25}
\begin{aligned}
x=\frac{C\mu(2\delta-z)^2}{2\delta(Cz+2\mu)},y=\frac{z(2\delta-z)}{2(Cz+2\mu)},z=\frac{2\mu}{1-\e^{\delta t}e_0}, \delta=\frac{-\mu}{C+1},A=\frac{2+C}{2}.
\end{aligned}
\end{equation} 
Similarly when $f=\lambda-\mu=\pm2z_1/z_0$, we obtain the following solutions:
\begin{subequations}
\label{3x-y3eq:28,31}
\begin{align}
\label{3x-y3eq:28}
&x=\frac{C (C+1) z^2}{4 \mu -2 C z},y=\frac{z (2 \mu +z)}{4 \mu -2 C z},z=\frac{-2\mu\e^{\mu t}}{\e^{\mu t}-e_0},\mu=\nu=\frac{\lambda}{2},A=\frac{C+2}{2},\\
\label{3x-y3eq:31}
&x=\frac{ACz(Az+2\lambda)}{2ACz+4\lambda},y=\frac{(C-1)A^2z^2}{2ACz+4\lambda},z=\frac{2\lambda\e^{\lambda t}}{A(e_0-\e^{\lambda t})},\lambda=\nu=\frac{\mu}{2},A=\frac{2}{2-C}.
\end{align}    
\end{subequations}

If $\lambda=\mu$, then $\gamma\geq1$.
If $\gamma\geq 2$, it follows that $g$ has zeros on $S_3$ and Lemma \ref{lem3} shows that $g=0$. If $g=0$, then $z'z_0+z^2-2z_1z=0$ and $z$ has no zeros unless $z=0$. Therefore, $\gamma=1$. Eliminating $y$ and $z$ from expressions for $f$, $g$ and $g'$ in terms of $x$, $y$ and $z$ shows that $a_1z^2+a_0z=0$, where $a_1=2C(\mu -\nu )(A C^2 \nu +(A-1)^2 (C+1) \mu),a_0=c_1g'+c_2g$. Here, $c_1\in\mathbb{C}[A,C],c_2\in\mathbb{C}[A,C,\mu,\nu]$. It follows from Theorem \ref{co:7} that $a_1=a_0=0$. If $A=0,C=-1,\mu\neq\nu$, then $a_0=0$ implies $g'=2g\nu$. Similarly, on eliminating variables between $f$, $f'$ and $g$, we obtain $g=\mu(\mu-\nu)$. Using (\ref{eq:30}), we express $x,y$ in terms of $z$, and substitute these into (\ref{eq:1,2}) yields $\mu\nu(\mu-\nu)=0$. Here, $\mu=0$ implies $g=0$. We discuss the case $\mu=\nu$ separately. When $\nu=0$, we obtain the solution: 
\begin{equation}
\label{l=mx-y3eq:1}
2z'=z^2-\mu^2, x=(z-\mu)^2(2z)^{-1}, y=z'z^{-1},\lambda=\mu,\nu=A=f=0,C=-1,g=\mu^2.
\end{equation}
If $A\neq0,\mu\neq\nu$, then $a_1=0$ implies $\nu=-(A-1)^2(C+1)A^{-1}C^{-2}\mu$.
Arguing similarly yields the following solution: $f=0,g=-\mu^2$ and,
\begin{equation}
\label{l=mx-y3eq:4}
2z'=(z+\mu)^2,x=-(\mu +z)^2(2z)^{-1},y=-(\mu +z)^2(2Cz)^{-1},A=1,\lambda=\mu,\nu=0.
\end{equation}
If $\lambda=\mu=\nu$, then $a_0=0$ shows that $g'=2g\mu$. Expressing $f$ and $f'$ in terms of $x$, $y$ and $z$ and eliminating $x$, we obtain $(A-1)z\{2C^2y+(1-A+2C-AC)z\}=0$.
So $A=1$; otherwise, we obtain $T(r,y)=T(r,z)$, which implies $N_3(r)=S(r,y)$, reducing the problem to section \ref{neq03*}. When $A=1$, we find that $f=f_0=0$. Using (\ref{eq:30}), we express $x,y$ in terms of $z$ and substituting into (\ref{eq:1}) shows that we have the solution: 
\begin{equation}
\label{l=mx-y3eq:6}
2z'=z(z+2\mu)-g,g'=2g\mu,x=(g-z^2)(2z)^{-1},y=C^{-1}x,\lambda=\mu=\nu,A=1.
\end{equation}
Explicit forms of (\ref{l=mx-yy-zz-xeq:2}) can be obtained by using (\ref{eq:22}).

From the above arguments, when $\gamma\neq1$, we conclude that $g\in\mathbb{C}$ which implies $z\in W$. Imposing $N_i(r)=S(r,z)$ for $i=1,2$ will not affect the proof. This is because when $z\in W$, then $N_i(r)=S(r,z)$ implies $N_i(r)=0,$ for $i\in\{0,1,2\}$ (see \cite{bank1984value} and \cite[Ch.1]{hayman1964meromorphic} for proof).

\subsection{$\mathbf{p_z},\mathbf{p_x}$ poles}
\label{neq0x-yy-z}
We find an entire function
\begin{equation}
\label{neq0x-yy-z3eq:1}
f=-x+Cy-ACz.
\end{equation}
Eliminating $x$ from $f$ and $f'$ results in a polynomial equation $R_{12}=0$, where $R_{12}$ is a polynomial in $z$ with coefficients polynomial in $y$, $f$ and $f'$. $R_{13}$ is defined similarly starting from $f$ and $f''$.
 Consider a polynomial division of $R_{13}$ by $R_{12}$ in $z$, i.e., $R_{13}=R_{12}Q+\mathfrak{r}$, where $Q$ is a polynomial in $z$ with coefficients containing $y$, and $\mathfrak{r}=b_1y+b_2z+b_3$. Here, $b_i\in\mathbb{C}[\lambda,\mu,\nu,C,f,f',f'']$. Since $R_{13}=R_{12}=0$, then $\mathfrak{r}=0$, which shows that $T(r,y)=T(r,z)+S(r,y)$. This implies $N_1(r)=S(r,y)$, unless $b_i=0$. When $\lambda$, $\mu$ and $\nu$ are pairwise distinct, the resonance condition is $A=B=-1/C$. The relation $b_1=0$ implies $C(C+1)f'=b_4f+b_5$, where $b_4,b_5\in\mathbb{C}[\lambda,\mu,\nu,C]$. Thus, we express $f',f''$ in terms of $f$ and substitute these into $b_3=0$ yielding $f^2+c_1f+c_0=0$ where $c_i$ are constants. This shows that $f$ is constant and $f'=f''=0$. Then, we can determine the value of $f$ and the relations between $\lambda,\mu,\nu$ from $b_i=0$.
Using $f,f'$, we express $y,z$ in terms of $x$. 
Moreover, (\ref{eq:1,1}) yields $x'=x(x+f+\lambda)$. Hence, we find the explicit forms of $x,y,z$. Substituting these $x,y,z$ into (\ref{eq:1}) allows us to verify whether the equations hold. Following this procedure, we find exactly one solution:
\begin{equation}
\label{l=my-zz-zeq:5}
x'=x(x+\nu),z=x(x-\breve{x}_0)^{-1},y=A(z-x),\mu=2\nu=2\lambda,B=A=-1/C.
\end{equation}

If $\lambda=\mu\neq\nu$, the resonance conditions are $A=-1/C,B=-1/C-\gamma$. Let, $g=z(x-\breve{x}_0)$, which is an entire function.
If $\gamma\geq2$, then $g=0$ by Lemma \ref{lem3}. Therefore, $\gamma=1$. Arguing similarly, we find that $\mathfrak{r}= b_2z+b_3$ where $b_i\in\mathbb{C}[C,\lambda,\nu,f,f']$. Thus, we obtain $b_i=0$. From $b_2=0$, we find that $(2C+1)f'=b_4f+b_5$ where $b_i\in\mathbb{C}[\lambda,\nu,C]$. Here, $C\neq-1/2$, since otherwise $D=0$. Then, $b_3=0$ yields $(2C+1)f^2(C+1)+c_1f+c_0=0$, where $c_i$ are constants. Hence, we conclude that $f$ is constant. From there, we argue similarly as in the case $\lambda\neq\mu\neq\nu$ and find the following solutions:
\begin{subequations}
\label{l=my-zz-zsubeq:5}
\begin{align}
\label{l=my-zz-zsubeq:5.1}
&x'=x\left(x+\frac{(C+1)\lambda}{2C+1}\right),y=\frac{-A^2(2 C+1) x^2}{(A-2)x-\lambda},z=\frac{C\lambda ^2(2 C+1)^{-1}}{(A-2)x-\lambda},\lambda=\mu,\nu=0,\\
\label{l=my-zz-zsubeq:5.2}
&x'=x\left(x+\frac{\nu}{B}\right),y=\frac{(Bx+\nu)^2}{C(B^2C x-\nu)},z=\frac{-C \nu ^2}{B^2 \left(B^2Cx- \nu \right)},\lambda=\mu=\frac{\nu }{B^2}.
\end{align}    
\end{subequations}

If $\lambda=\mu=\nu$, then we obtain $\gamma,\alpha\geq1$. If either $\gamma\geq2$ or $\alpha\geq2$, then Lemma \ref{lem3} shows that $xz=0$. Hence, $\alpha=\gamma=1,B=-(C+1)/C,$ $A=-1/(C+1)$. We may assume that $\lambda=0$ by using (\ref{eq:22}). There is a first integral \cite{lvsystem90}
\begin{equation}
\label{l=m=nx-yy-zeq:1}
h=x^2+C^2y^2+C^2(C+1)^{-2}z^2-2Cxy+2C(C+1)^{-1}zx+2C^2(C+1)^{-1}yz,
\end{equation}
which can be written as $2f'+(C+1)C^{-1}h-(C+1)C^{-1}f^2=0$. Hence, we can solve for $f$. We also have $f=-x+Cy-ACz,~ f'=2xz $, and, with (\ref{eq:1,1}), we obtain $(x+f/2)'-(x+f/2)^2+h/4=0$. Thus, $x,y,z$ can be found. Here, $x,y,z$ are meromorphic and belong to the case \ref{l=m=nx-yy-zz-x} (the explicit form has been shown in \cite{lvsystem90}). When $f=0$, then $f'=2xz=0$ and we have the 2D L-V case (section \ref{0-x-y=0}). The only solutions having just  $\mathbf{p_z},\mathbf{p_x}$ poles can be obtained as follows. We set $h=0,f\not\equiv0$ which implies $2f'-(C+1)C^{-1}f^2=0$ i.e. $f=-(C+1)\{2C(t-t_0)\}^{-1}$. Thus, we may find the closed forms of $x,y,z$. Setting $t_0=0$ and applying the inverse map of (\ref{eq:22}) yield 
\begin{equation}
\label{l=m=nx-yy-zeq:5}
{ 
\begin{aligned}
&x=\frac{-\lambda(M\e^{\lambda t}+C)}{(C+1)(M\e^{\lambda t}-1)},~z=\frac{\lambda C(M\e^{\lambda t}-1)}{M\e^{\lambda t}+C},~y=\frac{-(C+1)\lambda M^2\e^{2\lambda t}}{C(M\e^{\lambda t}-1)(M\e^{\lambda t}+C)},
\end{aligned}
}
\end{equation}
where $M\in\mathbb{C}$. From now on, we may dismiss the case when $\alpha=\gamma=1$.

\subsection{$\mathbf{p_z},\mathbf{p_x},\mathbf{p_0}$ poles and $\lambda\neq\mu,\mu\neq\nu$}
\label{neq0x-yy-z3}
The resonance condition is $A=B=-1/C$, so $C\neq-1$. The function $f$ defined in (\ref{neq0x-yy-z3eq:1}) satisfies $N(r,f)=N_0(r)$ and its Laurent series about any $t_0\in S_0$ is $f=f_0/(t-t_0)+f_1+\cdots$, where $f_i=-x_i+Cy_i-ACz_i$. Here, $x_i,y_i,z_i,i\leq1$ are known. Let
\begin{equation}
\label{neq0x-yy-z3eq:2}
g=f_0f'+f^2-2f_1f,
\end{equation}
which is entire, so $T(r,g)=S(r,y)$. If $2$ is a root of (\ref{eq:7}), then (\ref{eq:7}) yields $C^2+3C+1=0$ (vice versa).
If $C^2+3C+1\neq0$, then $x_2,y_2,z_2$ in (\ref{3x-y3eq:1}) are known. About $t_0\in S_0$, we have $g=3f_0f_2-f_1^2+O(t-t_0)$. So Lemma \ref{lem3} implies $g=3f_0f_2-f_1^2$; otherwise, $N_0(r)=S(r,y)$, which reduces the discussion to section \ref{neq0x-yy-z}. 
We first show that $x,y,z$ have a unique expansion about the $\mathbf{p_0}$ poles for $A=-1/C,B=-1/C-\gamma,\gamma\in\mathbb{Z}_0^+$. If (\ref{eq:7}) has no integer roots, we are done. If $N\in\mathbb{N}$ is a root of (\ref{eq:7}), then substituting (\ref{3x-y3eq:1}) into (\ref{neq0x-yy-z3eq:2}) yields the coefficient $t^{N-2}$, where $(x_0+N)x_N=P$, and $P$ can be expressed in terms of $x_i,y_i,z_i$ for $i<N$. If $N=-x_0$, then (\ref{eq:7}) shows that $C(\gamma+1)=0$ which is false. So $N\neq-x_0$ and $x_N$ is known. Since $x,y,z$ have no zeros, Lemma \ref{anslem}$(a),(b)$ shows that $x,y,z$ are non-elliptic in Class W and $l=1$.
About any $t_0\in S_1$, we express $\Breve{x}_i,\Breve{y}_i,\Breve{z}_i$ for $i\geq1$ in terms of $\Breve{z}_1$. We find that $g=\breve{g}_{0}+2\breve{g}_{1}(C\gamma+C+1)^{-1}(t-t_0)+\cdots$, where $\breve{g}_{0}=4\breve{z}_1^2+\Breve{a}_{01}\Breve{z}_1+\Breve{a}_{00},\breve{g}_{1}=\Breve{a}_{21}\breve{z}_1^2+\Breve{a}_{11}\Breve{z}_1+\Breve{a}_{10}$.
Here, $\breve{a}_{ij}$ are expressed in terms of $C,\gamma,\lambda,\mu,\nu$ only. Since $g$ is constant, then
\begin{equation}
    \label{l=mx-yy-z3eq:5}
    {\begin{aligned}
    4\breve{z}_1^2+\Breve{a}_{01}\Breve{z}_1+\Breve{a}_{00}=3f_0f_2-f_1^2,\quad \breve{g}_{1}=0
    \end{aligned}}
\end{equation}
which shows that $\Breve{z}_1$ has at most two choices or $m\leq2$. Arguing similarly yields $p\leq2$. In Lemma \ref{anslem}$(c),(d)$, the relations $V_x=0,V_y=0,V_z=0$ imply $p(C+1)+1=0,C^2+(C+1)(p+Cm+1)=0,C^2+C(C+1)m=0$. None of these two equations hold simultaneously for $m,p\leq2$. This shows that $x,y,z$ are simply-periodic and at most one of $x,y,z$ can take mid-form.
As $g=f_0f'+f^2-2f_1f$, then
\begin{equation}
    \label{3x-yy-z3eq:5}
    f=f_1-\sqrt{f_1^2+g}+2\sqrt{f_1^2+g}\e^{2f_0^{-1}\sqrt{f_1^2+g} t}\left(\e^{2f_0^{-1}\sqrt{f_1^2+g} t}-\e_0\right)^{-1},
\end{equation}
which gives $\delta=\pm2f_0^{-1}\sqrt{f_1^2+g}$ ($x,y,z$ are rational functions in $\tau=\e^{\delta t}$), and so
\begin{equation}
    \label{3x-yy-z3eq:6}
    \delta^2f_0^2=4(f_1^2+g).
\end{equation}
Lemma \ref{lemmamaxmin} below shows that $x,y,z$ cannot all take max-form or all take min-form.
So we only need to discuss the cases $x,y$ and $x,z$ both take min-form or max-form. Substituting (\ref{anslemeq:1}) into (\ref{3x-yy-z3eq:7}), the lowest degree terms from the expansion in $1/\tau$ yield
\begin{subequations}
    \label{3x-yy-z3eq:8}
    {\begin{align}
    &\lambda=\delta(m_x-M_x)-C\delta V_y\omega_{\infty,y}-\delta V_z\omega_{\infty,z},\\
    &\mu=\delta(m_y-M_y)+\delta V_z\omega_{\infty,z}C^{-1}-\delta V_x\omega_{\infty,x},\\
    &\nu=\delta(m_z-M_z)+\delta V_x\omega_{\infty,x}C^{-1}+\delta V_y\omega_{\infty,y},  
    \end{align}}
\end{subequations}
where $\omega_{\infty,u}=\max\{0,m_u-M_u+1\}$ for $u=x,y,z$. For each $m,p\in\{1,2\}$, we choose $m_x,m_y,m_z$ such that either $x,y$ or $x,z$ both take min-form or max-form. Substituting (\ref{3x-yy-z3eq:8}) into (\ref{3x-yy-z3eq:6}) yields $P_1(C)=0$ where $P_1\in\mathbb{C}[C]$. Using (\ref{l=mx-yy-z3eq:5}), we obtain $R_{\Breve{z}_1}(\breve{g}_{0}-3f_0f_2-f_1^2,\breve{g}_{1})=0$, leading to another relation $P_2(C)=0$ where $P_2\in\mathbb{C}[C]$. We cancel out the factors $C,(C+1)$ from $P_1,P_2$ which yields $\Tilde{P}_1(C)=\Tilde{P}_2(C)=0$. This checks which $C(C+1)\neq0$ allows $\breve{z}_1$ to satisfy (\ref{l=mx-yy-z3eq:5}). Here $C=-1$ implies $D=0$. Then we compute $R_{C}(\Tilde{P}_1,\Tilde{P}_2)$ to check whether it is $0$, which verifies if $\Tilde{P}_1,\Tilde{P}_2$ have common roots. However, for each $m,p\in\{1,2\}$, we obtain $R_{C}(\Tilde{P}_1,\Tilde{P}_2)\neq0$. So there are no solutions. We remark here even if $C^2+3C+1=0$, we still have $l=1,m,p\leq2$ as long as $g$ is constant. 

If $C^2+3C+1=0$, then $x_i,y_i,z_i,$ for $i\geq2$ in (\ref{3x-y3eq:1}) can be expressed in terms of $x_2$. First, consider the case $C=(-3-\sqrt{5})/2$. By (\ref{3x-y3eq:1}), we have, about each point $t_0\in S_0$, $g=g_{0}/2+(a_{11}x_2+a_{10})(t-t_0)/16-(9x_2^2+a_{21}x_2+a_{20})(t-t_0)^2/10+\cdots$ where $g_{0}=-3(5+\sqrt{5})x_2+a_{01}$ with $a_{ij}\in\mathbb{C}[\lambda,\mu,\nu]$.
If $g$ is constant, the above remark implies $x,y,z$ are simply-periodic in Class W with $l=1$ and $m,p\leq2$. Substituting (\ref{3x-yy-z3eq:8}) into (\ref{3x-yy-z3eq:6}), we determine $g$ from (\ref{3x-yy-z3eq:6}). However, we find that $R_{x_2}(g_0/2-g,a_{11}x_2+a_{10})\neq0$. This is false as $g=g_0/2,a_{11}x_2+a_{10}=0$. Thus, there are no solutions. If $g$ is not constant, Lemma \ref{lemmafurtherfunction} below shows that $N_0(r)=S(r,y)$, reducing the problem to section \ref{neq0x-yy-z}. A similar argument for $C=(\sqrt{5}-3)/2$ also leads to a contradiction. 
\counterwithin{lemma}{subsection}
\begin{lemma} 
\label{lemmamaxmin}
Suppose that $D\neq0$ and $x,y,z$ have no zeros. If $x,y,z$ all take max-form or all take min-form, then there are at most two types of singularities.     
\end{lemma}
\begin{proof}
By Remark \ref{remarkminmaxmid}, we may assume that $x,y,z$ all take max-form and $x$ has the least degree. Lemma \ref{anslem}$(c)$ shows that $x=c_x\tau^{M_x}\Pi_{i,j}(\tau-\tau_{i_j})^{-1}$ (similarly for $y,z$). Then $x=O(\tau^{n+p+l}),y=O(\tau^{m+p+l}),z=O(\tau^{m+n+l})$. Since $x$ has the least degree, then $f=O(\tau^{p+n+l})$. The function $f=-x+Cy-ACz$ has only the $\mathbf{p_y},\mathbf{p_0}$ poles (all poles are simple). Since $m(r,f)=S(r,y)$, then $f$ is a rational function in $\tau$ of degree $n+l$. This implies $n+p+l\leq n+l$ or $p=0$; otherwise $f$ is constant which contradicts Remark \ref{remarklinearatleast3}. Arguing similarly for $-y+Az-BAx$ shows that $n=0$.  
\end{proof}
\begin{lemma}
\label{lemmafurtherfunction}
Given an entire function $h$ such that $T(r,h)=S(r,x)$. Let
\begin{equation}
\label{furtherfunctioneq:1}
    H=a_1h^2+a_2h+a_3h''+a_4h',
\end{equation}
where $a_i$ are to be determined. Assume that for a fixed $i$ we have, about any $t_0\in S_i$,
\begin{equation}
\label{furtherfunctioneq:2}
    h=\sum_{k=0}^1a_{0k}\zeta^k+\sum_{k=0}^1a_{1k}\zeta^k(t-t_0)+\sum_{k=0}^2a_{2k}\zeta^k(t-t_0)^2+O((t-t_0)^3),
\end{equation}
where $a_{22}\neq0$ and $\zeta$ is a parameter. Then, there exists $a_{j}$ which are not all zero such that $h^k$ for some positive integer $k$ is the only dominant term in $H$ (Definition \ref{defdiffpoly}), and about any points in $S_i$, $H=H_0+O(t-t_0)$ where $H_0$ does not depend on $\zeta$. Also, if none of $x,y,z$ are entire, then either $N_i(r)=S(r,x)$ or $H,h\in\mathbb{C}$.
\end{lemma}
\begin{proof}
Assume that $N_i(r)\neq S(r,x)$.
If $a_{01}=0$, then we choose $H=a_2h$. Since all of $x,y,z$ have poles, then Lemmas \ref{lem2} and \ref{lem3} show that $H,h\in\mathbb{C}$. If $a_{01}\neq0$, then about each point $t_0$ in $S_i$, $H=\sum_{k=0}^2b_k\zeta^k+O(t-t_0)$ where $b_k$ are linear combinations of $a_k$. Setting $b_1=b_2=0$, we can find $a_k$ which are not all zero. As $T(r,h)=S(r,x)$, then $T(r,H)=S(r,x)$ and $H\in\mathbb{C}$ by Lemma \ref{lem3}. If $a_1\neq0$, then $h^2$ is dominant in $H$ and Theorem \ref{th:11} shows that $h$ is rational, so it is a polynomial as $h$ is entire. By balancing the degrees of $h,h',h''$ in (\ref{furtherfunctioneq:1}), we conclude $h\in\mathbb{C}$. If $a_1=0$, then $a_3=0$ since $a_{22}\neq0$. However, $b_1=b_2=0$ is a linear system in $a_1,\dots,a_4$, so at least there has to be two free parameters i.e. $a_2,a_4$ are free, and we may set $a_4=0$.
\end{proof}
\begin{remark}
\label{remarkfurtherfunction}    
We will also apply Lemma \ref{lemmafurtherfunction} with the analogues of (\ref{furtherfunctioneq:1}),(\ref{furtherfunctioneq:2}) whose proofs are similar to Lemma \ref{lemmafurtherfunction}. We will apply these analogues of Lemma \ref{lemmafurtherfunction}. 
\end{remark}

\subsection{$\mathbf{p_z},\mathbf{p_x},\mathbf{p_y}$ poles and $\lambda,\mu,\nu$ pairwise distinct}
\label{neq0x-yy-zz-x}
We have the resonance conditions $BC=AC=AB=-1$. By using (\ref{conjugatemap}), we choose $A=B=C=\mathrm{i}$. We denote $\sum_{cyc}G(x,y,z,\lambda,\mu,\nu)=G(x,y,z,\lambda,\mu,\nu)+G(y,z,x,\mu,\nu,\lambda)+G(z,x,y,\nu,\lambda,\mu).$
We find an entire function
\begin{equation}
\label{eq:38}
{\begin{aligned}
f=xyz+\sum_{cyc}\left\{a_{xy}xy+\frac{\mathrm{i}+1}{2}\left[(2-\mathrm{i})\lambda\mu+\mu^2+(2\mathrm{i}-1)\nu\lambda+\mathrm{i}\nu^2\right]x+2\mathrm{i}\mu\nu x\right\},
\end{aligned}}
\end{equation}
where $a_{xy}=\lambda-\mathrm{i}\mu+(\mathrm{i}-1)\nu$. Expressing $\hat{x}_i$, $\hat{y}_i$ and $\hat{z}_i$ in (\ref{3x-y3eq:3}) in terms of $\hat{x}_1$, and substituting (\ref{3x-y3eq:3}) into (\ref{eq:38}) yields
\begin{equation}
\label{3x-yy-zz-xeq:1} 
{\begin{aligned}
f=2a_{xy}\hat{x}_1^2+a_{01}\hat{x}_1+a_{00}+(2+2\mathrm{i})(\lambda-\mu)\sum_{i=1}^2\sum_{j=0}^{i+2}a_{ij}\hat{x}_1^j(t-t_0)^i+\cdots
\end{aligned}
}
\end{equation}
where $a_{13}=2/3$, $a_{24}=\mathrm{i}$, and the rest of $a_{ij}$ depend only on $\lambda,\mu,\nu$. If $a_{xy}\neq0$, we define $F=a_1f^3+a_2f''+a_3(f')^2+a_4ff'+a_5f^2+a_6f'+a_7f$. Then, $F$ and (\ref{3x-yy-zz-xeq:1}) are analogues of Lemma \ref{lemmafurtherfunction}. Thus, by Remark \ref{remarkfurtherfunction}, $F,f\in\mathbb{C}$. Since $a_{13}=2/3$, then $\hat{x}_1$ has at most three choices. Thus, Lemma \ref{anslem}$(a),(d)(i)$ and \ref{subanslem}$(b)$ deduce that $x,y,z$ are simply-periodic in Class $W$ and $p\leq3$.
Similar arguments show that $m,n\leq3$. We consider $x,y,z$ as rational functions in $\tau=\e^{\delta t}$. As $V_xV_yV_z=(\mathrm{i}n-p)(\mathrm{i}p-m)(\mathrm{i}m-n)\neq0$, then $x,y,z$ cannot take mid-form by Lemma \ref{anslem}$(c)(i)$. By Lemma \ref{lemmamaxmin} and Remark \ref{remarkminmaxmid}, we only need to discuss when $x$ takes min-form and $y,z$ take max-form. 
The forms of $x,y,z$ are given in Lemma \ref{anslem}$(c)$.
Using the information in Lemma \ref{anslem}$(c)$, the expansions about infinity of $x,y,z$ are $x=O(\tau^{-(n+p+l)}),y=\acute{y}_0+\acute{y}_1\tau^{-1}+\cdots,z=\acute{z}_0+\acute{z}_1\tau^{-1}+\cdots$, where $\acute{y}_0=\delta V_y,\acute{z}_0=\delta V_z$. Of course, in this case $l=0$. We substitute these into (\ref{3x-yy-z3eq:7.1}) and (\ref{3x-yy-z3eq:7.3}),
and equate the coefficients of $\tau^{-j}$ for $0\leq j\leq n+p+l-1$. For $j=0$, we obtain $\mu=-A\acute{z}_0,\nu=-\acute{y}_0$.
For $1\leq j\leq n+p+l-1$, we obtain    
    \begin{subequations}
    \label{3x-yy-zz-xeq:12}
        \begin{align}
        \label{3x-yy-zz-xeq:12.1}
        &-\delta j\acute{y}_j-A\acute{y}_0\acute{z}_j=Q_j,\\
        \label{3x-yy-zz-xeq:12.2}
        &-\acute{z}_0\acute{y}_j-\delta j\acute{z}_j=R_j,
    \end{align}
    \end{subequations}
where $Q_j,R_j$ are sums of monomials consisting of $\acute{y}_k,\acute{z}_k,$ for $k=1,\dots,j-1$, without constant terms, and $Q_1=R_1=0$. Consider (\ref{3x-yy-zz-xeq:12}) as a linear system in $\acute{y}_j,\acute{z}_j$, then the determinant of the coefficient matrix is $D(j)=\delta^2(j^2-m^2-np+\mathrm{i}m(p-n))$. We observe that for each $m,n,p\in[1,3]$, $D(j)\neq0$. 
Thus, $\acute{y}_j=\acute{z}_j=0$ for $j=1,\dots,n+p+l-1$ and $x=O(\tau^{n+p+l}),y=\acute{y}_0+O(\tau^{n+p+l}),z=\acute{z}_0+O(\tau^{n+p+l})$. We note that the function
\begin{equation}
\label{3x-yy-zz-xeq:0012}    
I:=-A^{-1}y+z-Bx
\end{equation}
has only $\mathbf{p_z},\mathbf{p_0}$ poles (all poles are simple). Since $m(r,I)=S(r,y)$, then it is a rational function in $\tau$ of degree $p+l$. However, $I=\acute{z}_0-A^{-1}\acute{y}_0+O(\tau^{n+p+l})$ which implies $n+p+l\leq p+l$ or $n=0$. Therefore, there are no solutions here.
    
\subsection{$\mathbf{p_z},\mathbf{p_x},\mathbf{p_y},\mathbf{p_0}$ poles and $\lambda,\mu,\nu$ pairwise distinct}
\label{neq0x-yy-zz-x3}
Similar to section \ref{neq0x-yy-zz-x}, we choose $A=B=C=\mathrm{i}$. We find that (\ref{eq:7}) has no roots in $\mathbb{N}$, so $x,y,z$ have a unique expansion about any $t_0\in S_0$. Thus, Lemmas \ref{anslem}$(a),(d)(i)$ and \ref{subanslem}$(b)$ show that $x,y,z$ are simply-periodic in Class $W$ with $l=1$. Since $V_x=\mathrm{i}n-p+(\mathrm{i}-1)/2\neq0$, Lemma \ref{anslem}$(c)$ shows that $x$ cannot take mid-form. The same reasoning applies to $y$ and $z$. Similar to section \ref{neq0x-yy-zz-x}, we only need to consider the case where $x$ takes min-form and $y,z$ take max-form. The determinant of (\ref{3x-yy-zz-xeq:12}) is $D(j)=\delta^2(j^2-pn-m^2-m-(n+p+1)/2+\mathrm{i}(p-n)(m+1/2))$. We see that $D(j)=0$ implies $n=p,\ j^2=p^2+m^2+p+m+1/2$, which is false. Thus, $D(j)\neq0$ and the argument in section \ref{neq0x-yy-zz-x} shows that there are no solutions here.

\subsection{$\mathbf{p_z},\mathbf{p_x},\mathbf{p_0}$ poles and $\lambda=\mu\neq\nu$}
\label{l=mx-yy-z3}
The resonance conditions are $A=-1/C,B=-1/C-\gamma$, so $B\neq1$. When $2$ is a root of (\ref{eq:7}), then (\ref{eq:7}) yields $\psi:=C(3-2\gamma-\gamma^2)+C^2(1+\gamma)^2-\gamma+1=0$ (vice versa). 

First, we discuss the case $\psi\neq0$. Let $f,g$ be defined as in (\ref{neq0x-yy-z3eq:1}) and (\ref{neq0x-yy-z3eq:2}) in section \ref{neq0x-yy-z3}. Arguing similarly yields $g=3f_0f_2-f_1^2,$ $f$ is non-elliptic in Class W and $l=1$. The same arguments there imply $m\leq2$ where we also obtain (\ref{l=mx-yy-z3eq:5}) with $\breve{a}_{ij}$ depend only on $C,\gamma,\lambda,\nu$. The argument for $p\leq2$ will not work here as there are two parameters for the $\mathbf{p_z}$ poles. We discuss the case where $f$ is rational at the end. Then $f,x,y,z$ are simply-periodic in Class W. We also have (\ref{3x-yy-z3eq:5}), then $l=1$ shows that $\delta=\pm2\sqrt{f_1^2+g}/f_0$ and (\ref{3x-yy-z3eq:6}).
Lemma \ref{anslem}$(c)$ shows that $\gamma p+m_z\leq m+1$. So we may find all possible values of $m,p,\gamma,m_z$. Ans\"atze of $x,y,z$ are given in (\ref{anslemeq:01}) and (\ref{anslemeq:001}).
Let $\omega_{\infty,z}=\max\{0,\gamma p+m_z-M_z+1\},\omega_{\infty,u}=\max\{0,m_u-M_u+1\}$ where $u=x,y$.
Substituting the ans\"atze of $x,y,z$ into (\ref{3x-yy-z3eq:7}), then the lowest degree terms from the expansion in $1/\tau$ yield
\begin{subequations}
\label{l=mx-yy-zz-xeq:4}
    {\begin{align}
    \label{l=mx-yy-zz-xeq:4.1}
    &\delta(m_x-M_x)=C\delta V_y\omega_{\infty,y}+\acute{z}_0\omega_{\infty,z}+\lambda,\\
    \label{l=mx-yy-zz-xeq:4.2}
    &\delta(m_y-M_y)=-C^{-1}\acute{z}_0\omega_{\infty,z}+\delta V_x\omega_{\infty,x}+\lambda,\\
    \label{l=mx-yy-zz-xeq:4.3}
    &\delta(\gamma p+m_z-M_z)=-(\gamma C+1)C^{-1}\delta V_x\omega_{\infty,x} +\delta V_y\omega_{\infty,y}+\nu.    
    \end{align}}
\end{subequations}
If $C=-1$, then from (\ref{l=mx-yy-zz-xeq:4.1}),(\ref{l=mx-yy-zz-xeq:4.2}), we have
\begin{equation}
    \label{l=mx-yy-z3eq:17}
    m_x+V_y\omega_y=(m_y-m)-V_x\omega_x.
\end{equation}
For each $m,p,\gamma$, we find all the possible values of $m_x,m_y$ from (\ref{l=mx-yy-z3eq:17}). However, the forms of $x,y,z$ show that at least one of $\tau_0,\tau_{i_j}$ vanishes.  
Hence, $C\neq-1$, then (\ref{l=mx-yy-zz-xeq:4}) implies
\begin{equation}
\label{l=mx-yy-z3eq:8}
    {\begin{aligned}
    &\lambda=\delta(m_x-M_x+Cm_y-CM_y-CV_y\omega_{\infty,y}-CV_x\omega_{\infty,x})(C+1)^{-1},\\
    &\nu=\delta\{\gamma p+m_z-M_z-(C\gamma+1)C^{-1}V_x\omega_{\infty,x}-V_y\omega_{\infty,y}\} .
    \end{aligned}}
\end{equation}
If $x,y$ do not take the mid-form,we find $m_x,m_y$ for each $m,p,\gamma,m_z$ corresponding to whether $x,y$ take min-form or max-form. We argue similarly as in section \ref{neq0x-yy-z3} to obtain polynomial relations $P_1(C)=P_2(C)=0$. In $P_1,P_2$, we cancel out the factors $C,C(\gamma+1)+1,(C\gamma+1)(C+1)+C^2,(1-\gamma)(3C+1+C\gamma)+C^2(1+\gamma)^2$, as one of these vanishing implies $C(1-B)y_0\psi=0$. This yields the relations $\Tilde{P}_1(C)=\Tilde{P}_2(C)=0$. We then compute $R_{C}(\Tilde{P}_1,\Tilde{P}_2)$ to see whether $\Tilde{P}_1,\Tilde{P}_2$ have common roots. We check the ans\"atze of the cases where $R_{C}(\Tilde{P}_1,\Tilde{P}_2)=0$. We find six cases of $\mathbf{w}=(m_x,m_y,m,p,\gamma,m_z,C)$ which yield the solutions (express in terms of $F_i$) such that $\lambda\neq\nu$ as follows:
\begin{enumerate}
    \item $\mathbf{w}=(2,0,1,1,2,0,C_1)$, where $27C_1^3+33C_1^2+12C_1+1=0$ yields $F_1=(3C_1+1)\tau+\hat{\tau}_{1},~F_2=(3C_1+1)^2\tau+3(2C_1+1)\hat{\tau}_{1},\hat{\tau}_{1}\in\mathbb{C}$ and
    \begin{equation}
    \label{l=mx-yy-z3eq:11}
    {\begin{aligned}
    &y=\frac{-3(3C_1+2)(2C_1+1)^2\delta\hat{\tau}_{1}^3}{C_1(3C_1+1)^2(\tau-\hat{\tau}_{1})F_1F_2},~z=\frac{-C_1(33C_1^2+21C_1+1)\delta(\tau-\hat{\tau}_{1})^2}{3(15C_1^2+15C_1+4)F_1F_2},\\&x=\frac{-(3C_1+2)\delta\tau^2}{(\tau-\hat{\tau}_{1})F_1},~(\lambda,\nu)=\left(\frac{-C_1(6C_1+1)\delta}{3C_1^2+4C_1+1},\frac{-(6C_1^2+7C_1+2)\delta}{C_1(3C_1+1)}\right).
    \end{aligned}}   
    \end{equation}
    \item  $\mathbf{w}=(3,0,1,2,1,0,C_2)$, where $6C_2^2+6C_2+1=0$, yields $F_1=(\tau-\hat{\tau}_{1})(\tau-\hat{\tau}_{2}),~F_2=(2C_2+1)^2\tau-(2C_2^2-1)(\hat{\tau}_{1}+\hat{\tau}_{2}),F_3=(2C_2+1)\tau+(C_2+1)(\hat{\tau}_{1}+\hat{\tau}_{2}),$ $(3C_2+2)(\hat{\tau}_{1}^2+\hat{\tau}_{2}^2)+\hat{\tau}_{1}\hat{\tau}_{2}=0,\hat{\tau}_{j}\in\mathbb{C}$, $\lambda=(24C_2^2-1)\delta,\nu=(15-24C_2^2)\delta$, and
    \begin{equation}
    \label{l=mx-yy-z3eq:12}
    {
    \begin{aligned}
        y=\frac{-(33C_2+26)\delta\hat{\tau}_{1}^2\hat{\tau}_{2}^2}{3(3C_2+1)^2F_1F_2F_3},~
        z=\frac{-(6C_2+2)\delta F_1}{18(5C_2+4)F_2F_3},~x=\frac{-(4C_2+3)\delta\tau^3}{F_1F_3}.
    \end{aligned}}
    \end{equation}
    \item $\mathbf{w}=(2,0,2,1,3,0,-3/8)$, yields $\lambda=9\nu=3\delta,F_1=3\tau^2-6\tau\hat{\tau}_{1}+4\hat{\tau}_{1}^2,\hat{\tau}_{1}\in\mathbb{C}$ and
    \begin{equation}
    \label{l=mx-yy-z3eq:13}
    {\begin{aligned}
     x=\frac{\delta\tau^2}{(\tau-\hat{\tau}_{1})(\tau-2\hat{\tau}_{1})},~y=\frac{8\delta\hat{\tau}_{1}^4}{3(\tau-\hat{\tau}_{1})(\tau-2\hat{\tau}_{1})F_1},~z=\frac{9\delta(\tau-\hat{\tau}_{1})^3}{(\tau-2\hat{\tau}_{1})F_1}.    
    \end{aligned}}    
    \end{equation}
    \item $\mathbf{w}=(0,3,1,1,2,0,C_1),(0,4,1,2,1,0,C_2),(0,4,2,1,3,0,-3/8)$; solution can be obtained by changing $\delta\to-\delta$ in (\ref{l=mx-yy-z3eq:11}),(\ref{l=mx-yy-z3eq:12}),(\ref{l=mx-yy-z3eq:12}) respectively.
\end{enumerate}
Next we discuss when either $x$ or $y$ takes mid-form. By Lemma \ref{anslem}(c), we obtain $V_x=0$ or $V_y=0$ which can be simplified to $Cp(\gamma+1)+p+1=0,$ or $C^2(m+1)(\gamma+1)+Cm+(p+1)(C\gamma+C+1)=0$. Here $V_x=V_y=0$ cannot hold simultaneously for each value of $m,p,\gamma,m_z$. So only one of $x,y$ can take the mid-form. If $x$ takes mid-form, we find $C$ from $V_x=0$. Then we substitute (\ref{l=mx-yy-z3eq:8}), $C$ into (\ref{3x-yy-z3eq:6}). However, the identity fails. In the case where $y$ takes the mid-form, we argue similarly and obtain $(m,p,\gamma,m_z)=(1,1,1,0),(1,1,1,1)$ which satisfy (\ref{3x-yy-z3eq:6}).
If $(m,p,\gamma,m_z)=(1,1,1,0)$, then $m_x=0,m_y=1$ and $4C^2+5C+2=0$. Using (\ref{conjugatemap}), we choose $C=\mathrm{i}(5\mathrm{i}-\sqrt{7})/8$. Then, we obtain: $\lambda=2\nu=-2\delta,F_1=4\mathrm{i}\tau+(\sqrt{7}-\mathrm{i})\hat{\tau}_{1},F_2=4\tau+(\mathrm{i}\sqrt{7}-3)\hat{\tau}_{1},\hat{\tau}_{1}\in\mathbb{C}$, and
    \begin{equation}
    \label{l=mx-yy-z3eq:16}
    {\begin{aligned} 
    z=\frac{4(\sqrt{7}-5\mathrm{i})\delta(\tau-\hat{\tau}_{1})\hat{\tau}_{1}}{F_1F_2},~x=\frac{-(3\mathrm{i}+\sqrt{7})\delta\hat{\tau}_{1}^2}{(\tau-\hat{\tau}_{1})F_1},~y=\frac{16\mathrm{i}\delta\tau\hat{\tau}_{1}^2}{(\tau-\hat{\tau}_{1})F_1F_2}.    
    \end{aligned}}
    \end{equation}
The case $(m,p,\gamma,m_z)=(1,1,1,1)$ can be obtained from (\ref{l=mx-yy-z3eq:16}) by changing $\delta\to-\delta$. We remark here even if $\psi=0$, we still obtain $l=1,m\leq2$ as long as $g$ is constant.

Now we assume that $\psi=0$. Here, $\gamma=1$ iff $C=-1$. We can express $x_i,y_i,z_i,$ for $i\geq2$ in terms of $x_2$. We now show that $g$ is constant and $x_2$ is unique. If $C\neq-1$, then about each point $t_0\in S_0$ (we may assume $t_0=0$),
\begin{equation}
    \label{l=mx-yy-z3eq:18}
    {\begin{aligned}
    g=\frac{3C^3(\gamma+1)^2x_2}{16P_2^2P_{3}}+a_{00}+\sum_{i=0}^1a_{1i}x_2^it+\left(\frac{P_1x_2^2}{(\gamma-1)^4P_2^6P_{3}^2}+\sum_{i=0}^1a_{2i}x_2^i\right)t^2+O(t^3),
    \end{aligned}}
\end{equation}
where $a_{ij}$ depends only on $C,\gamma,\lambda,\nu$ and $P_1,P_2,P_3\in\mathbb{C}[C,\gamma]$.
We find that $R_C(P_1,\psi)=0$ has no roots in $\mathbb{N}$, so $P_1\neq0$. Applying Lemma \ref{lemmafurtherfunction} with $h=g$ shows that $g\in\mathbb{C}$. If $C=-1$, the analogue of (\ref{l=mx-yy-z3eq:18}) is $g=\left((\lambda-2\nu)^2/2-6x_2\right)\left(1-96\nu t-1440(5\nu-\lambda)t^2\right)+O(t^3)$.
Assume that $12x_2\neq(\lambda-2\nu)^2$. Using Lemma \ref{lem3} again yields $g'+96\nu g=0$. Then, the expansion becomes $g'+96\nu g=\left((\lambda-2\nu)^2/2-6x_2\right)(g_0+g_1t)+O(t^2)$ where $g_i$ are expressed in term of $\lambda,\nu$ only. So $g_0=g_1=0$. Simplifying these yields $\lambda=\nu=0$; a contradiction. Therefore, $g\in\mathbb{C}$ in any cases. Hence, as in the case $\psi\neq0$, we conclude that $x,y,z$ are in Class W with $l=1$ and $m\leq2$. The case where $x,y,z$ are rational is discussed at the end of this section. If $C\neq-1$, we substitute (\ref{l=mx-yy-z3eq:8}) into (\ref{3x-yy-z3eq:6}) to get the value of $g$. Using (\ref{l=mx-yy-z3eq:18}), we determine the value of $x_2$, and consequently, $f_2$ is also known. Then using (\ref{l=mx-yy-z3eq:5}), we obtain 
\begin{equation}
\label{l=mx-yy-z3eq:19}
R_{\breve{z}_1}(4\breve{z}_1^2+\breve{a}_{01}\breve{z}_1+\breve{a}_{00}-3f_0f_2-f_1^2,\breve{g}_1)=0.    
\end{equation}
For each $\gamma,m,p$, we determine $C$ from $\psi=0$. However, we find that (\ref{l=mx-yy-z3eq:19}) fails to hold for those values $\gamma,m,p,C$. When $C=-1$, we find that (\ref{l=mx-yy-z3eq:17}) fails to hold for each $m,p,\gamma$. Therefore, no solutions exist. 

If $f$ is rational, then $x,y,z$ are also rational. In both cases $\psi\neq0$ and $\psi=0$, the above discussions show that $l=1$ and $m\leq2$. Lemma \ref{anslem}$(d)$ yields $\nu=0,d_z=-\lambda,C=-1,\gamma p=m+1,V_x=V_y=0$. However, none of $m,p,\gamma$ make $V_x=V_y=0$ hold.

\subsection{$\mathbf{p_z},\mathbf{p_x},\mathbf{p_y}$ poles and $\lambda=\mu\neq\nu$}
\label{l=mx-yy-zz-x}

We have the resonance conditions $B=C,A=-1/C,C^2+\gamma C+1=0$. Then $C=-1$ iff $C\in\mathbb{Q}$ which is equivalent to $\gamma=2$. If $\gamma\geq2$, we find an entire function
\begin{equation}
    \label{l=mx-yy-zz-xeq:1}
    {\begin{aligned}
    f=&x^2+C^2y^2+(C+1)\{(2\lambda-\nu)(Cy-x)(C-1)^{-1}-C^{-1}xz+yz\}-2Cxy\\&+{\{(C+1)^2\lambda-(C^2+1)\nu\}}\{C(C-1)\}^{-1}z-{(C-1)(\lambda-\nu)^{-1}xyz}.  
    \end{aligned}}
\end{equation}
Assume that $C\neq-1$. By expressing $\Breve{x}_i,\Breve{y}_i,\Breve{z}_i$ in terms of $\Breve{z}_1$, we have about $t_0\in S_1$  
\begin{equation}
\label{l=mx-yy-zz-xeq:2}
    {\begin{aligned}
    f=2(C+1)\breve{z}_1^2+\sum_{i=0}^1a_{0i}\breve{z}_1^i+\frac{2(C^2-1)}{3C}\sum_{j=1}^2\sum_{i=0}^{j+2}a_{ji}\breve{z}_1^i(t-t_0)^j +O((t-t_0)^3),    
    \end{aligned}}
\end{equation}
where $a_{13}=2,a_{24}=1$ and other $a_{ji}$ are expressed in terms of $C,\lambda,\nu$ only. Applying Remark \ref{remarkfurtherfunction} on $F=a_1f^3+a_2f''+a_3(f')^2+a_4ff'+a_5f^2+a_6f'+a_7f$ and (\ref{l=mx-yy-zz-xeq:2}) asserts that both $F,f$ are constant. Then, (\ref{l=mx-yy-zz-xeq:2}) shows that $\breve{z}_1$ has at most two choices. 
When $C=-1$, equation (\ref{l=mx-yy-zz-xeq:2}) becomes
\begin{equation}
    \label{l=mx-yy-zz-xeq:17}
    {\begin{aligned}
    f=\sum_{j=0}^2\sum_{i=0}^{j+1}b_{ji}\breve{z}_1^i(t-t_0)^j+O((t-t_0)^3) ,
    \end{aligned}}
\end{equation}
where $b_{01}=-(\lambda+5\nu)/2,b_{12}=-(3\lambda+\nu)/3,b_{23}=(3\nu-11\lambda)/12$ and others $b_{ji}$ depend only on $a_{ji}$.
If $b_{01}=0$, then $f$ is constant. If $b_{12}=0$, then applying Remark \ref{remarkfurtherfunction} on $F=a_1f''+a_2f^3+a_3f^2+a_4f$ and (\ref{l=mx-yy-zz-xeq:17}) yields that $f$ is constant. If $b_{23}=0$, then $F=a_1f^2+a_2f+a_3f'+a_4f''$ also imples that $f$ is constant. Additionally, $b_{01}=b_{12}=0$ iff $\lambda=\nu=0$, so either $b_{01}$ or $b_{12}$ is non-zero i.e. $\Breve{z}_1$ has at most two choices. If $\gamma=1$, then by using (\ref{conjugatemap}), we choose $C=-(1+\mathrm{i}\sqrt{3})/2$. We find an entire function 
\begin{equation}
\label{l=mx-yy-zz-xeqsub:1}
    {\begin{aligned}
    g&=2x^3-2y^3+(3+3\mathrm{i}\sqrt{3})(x^2y-xz^2+\mathrm{i}xy^2+yz^2)+(3-3\mathrm{i}\sqrt{3})\{x^2z+xyz\\&+y^2z-(\lambda-\nu)(2yz+\mathrm{i}\sqrt{3}z^2-6xy)\}-6\mathrm{i}\sqrt{3}(\lambda-\nu)x^2
    -(3\mathrm{i}\sqrt{3}+3)(\lambda-\nu)\{2xz\\&-3y^2+\mathrm{i}(\lambda-2\nu)z\}
    -3(2\lambda^2-3\lambda\nu+\nu^2)\{2x+(\mathrm{i}\sqrt{3}+1)y\}.
    \end{aligned}}
\end{equation}
The expansion of $g$ about each point $t_0\in S_1$ is 
\begin{equation}
    \label{l=mx-yy-zz-xeqsub:2}
    {\begin{aligned}
    g=g_0+g_1(t-t_0)+\frac{2}{5}(\mathrm{i}\sqrt{3}-1)(\lambda-\nu)\sum_{j=2}^4\sum_{i=0}^{j+2}c_{ji}\Breve{z}_1^i(t-t_0)^j+\cdots,
    \end{aligned}}
\end{equation}
where $g_0=-16\breve{z}_1^3+\sum_{i=0}^2c_{0i}\Breve{z}_1^i,~
    g_1=\sum_{i=0}^3c_{1i}\Breve{z}_1^i,$ with $c_{24}=15,c_{35}=2,$ $c_{46}=1+\mathrm{i}\sqrt{3}$ and other $c_{ji}$ depend only on $\lambda,\nu$.
Applying Remark \ref{remarkfurtherfunction} on $G=a_1g'''g''+a_2(g'')^2+a_3g''g+a_4g^2+a_5g'''+a_6g''+a_7g'+a_8g+a_{9}g^{(4)}+a_{10}g^3$ and (\ref{l=mx-yy-zz-xeqsub:2}) shows that $g$ is constant. By considering $g_1,g_0-g$ as polynomials in $\Breve{z}_1$, we find that the remainder of the division of $g_1$ by $g_0-g$ is $\mathfrak{r}=c_2 \Breve{z}_1^2+c_1\Breve{z}_1+c_0$ (expressions of $c_i$ are not given here). Since $g_1=g-g_0=0$, then $\mathfrak{r}=0$. We find that both $c_2,c_1$ cannot be zero otherwise $\lambda=\nu=0$. 
Therefore, $\breve{z}_1$ has at most two choices in any cases. Arguing similarly shows that $\check{x}_1$ also has at most two choices. Hence, $x,y,z$ are non-elliptic in Class W by Lemma \ref{anslem}$(a),(b)$, and $m,n\leq2$. We first consider the case when $x,y,z$ are simply-periodic. Since $V_x=-(p+n/C),V_y=-(m+p/C)$, Lemma \ref{anslem}$(c)$ shows that $x$ takes mid-form iff $C=-1,n=p$, and $y$ takes mid-form iff $C=-1,m=p$.
Ans\"atze for $x,y,z$ are given in (\ref{anslemeq:01}) and (\ref{anslemeq:001}) where $l=0,$ and $\gamma p+m_z\leq m+n$. 

Assume $C\neq-1$. 
Using Lemma \ref{anslem}$(c)$, then the expansions of $x,y,z$ at $0$ are 
\begin{equation}
\label{l=mx-yy-zz-xeq:3,1}
x=\tau^{m_x}(\ddot{x}_0+\ddot{x}_1\tau+\cdots),y=\tau^{m_y}(\ddot{y}_0+\ddot{y}_1\tau+\cdots),z=\tau^{m_z}(\ddot{z}_0+\ddot{z}_1\tau+\cdots).
\end{equation}
The expansions of $x,y,z$ at infinity are
\begin{equation}
\label{l=mx-yy-zz-xeq:3,2}
x=\tau^{m_x-M_x}(\acute{x}_0+\acute{x}_1\tau+\cdots),y=\tau^{m_y-M_y}(\acute{y}_0+\acute{y}_1\tau+\cdots),z=\tau^{\gamma p+m_z-M_z}(\acute{z}_0+\acute{z}_1\tau+\cdots).
\end{equation}
Similar argument as in section \ref{l=mx-yy-z3} yields (\ref{l=mx-yy-zz-xeq:4}). Let $\omega_{0,u}=\max\{0,1-m_u\}$ where $u=x,y,z$.
The analogues of (\ref{l=mx-yy-zz-xeq:4}) for (\ref{l=mx-yy-zz-xeq:3,1}) are 
    \begin{subequations}
    \label{l=mx-yy-zz-xeq:04}
    \begin{align}
    \label{l=mx-yy-zz-xeq:04.1}
    &\delta m_x=C\ddot{y}_0\omega_{0,y}+\ddot{z}_0\omega_{0,z}+\lambda,\\
    \label{l=mx-yy-zz-xeq:04.2}
    &\delta m_y=-C^{-1}\ddot{z}_0\omega_{0,z}+\ddot{x}_{0}\omega_{0,x}+\lambda,\\
    \label{l=mx-yy-zz-xeq:04.3}
    &\delta m_z=C\ddot{x}_0\omega_{0,x}+\ddot{y}_0\omega_{0,y}+\nu.
    \end{align}
    \end{subequations}
If $x,y$ both take the min-form, then $m_x=m_y=0,\ddot{x}_0=-\delta V_x,\ddot{y}_0=-\delta V_y$.
If $m_z\neq0$, then (\ref{l=mx-yy-zz-xeq:04.1}) and (\ref{l=mx-yy-zz-xeq:04.2}) imply $C\ddot{y}_0=\ddot{x}_0$ which can be simplified as $C\gamma m+m+n=0$ i.e. $m=n=0$ which is false. Thus, $m_z=0$ and (\ref{l=mx-yy-zz-xeq:04}) yields
\begin{equation}
\label{l=mx-yy-zz-xeq:004}
    \lambda=-C(\ddot{x}_0+\ddot{y}_0)(C+1)^{-1},~\ddot{z}_0=C(\ddot{x}_0-C\ddot{y}_0)(C+1)^{-1},~\nu=-C\ddot{x}_0-\ddot{y}_0.
\end{equation}
Substituting (\ref{l=mx-yy-zz-xeq:3,1}) into (\ref{3x-yy-z3eq:7}) and equating the coefficients of $\tau^j$ yield
    \begin{subequations}
    \label{l=mx-yy-zz-xeq:9}
    {\begin{align}
    \label{l=mx-yy-zz-xeq:9.1}
    &j\delta\ddot{x}_j-C\ddot{x}_0\ddot{y}_j-\ddot{x}_0\ddot{z}_j=P_j,\\
    \label{l=mx-yy-zz-xeq:9.2}
    &\ddot{y}_0\ddot{x}_0-j\delta\ddot{y}_j+C^{-1}\ddot{y}_0\ddot{z}_j=Q_j,\\
    \label{l=mx-yy-zz-xeq:9.3}
    &-C\ddot{z}_0\ddot{x}_j-\ddot{z}_0\ddot{y}_0+j\delta\ddot{z}_j=R_j,
    \end{align}}
    \end{subequations}
    where $P_j,Q_j,R_j$ are the sum of the monomials in $\ddot{x}_1,\ddot{y}_1,\ddot{z}_1,\dots,\ddot{x}_{j-1},\ddot{y}_{j-1},\ddot{z}_{j-1}$, without constant terms, and $P_1=Q_1=R_1=0$. The linear system (\ref{l=mx-yy-zz-xeq:9}) has the determinant $\tilde{D}(j)=(Cd_{1j}+d_{2j})\delta^3B^{-3}$ where $d_{ij}=\sum_{s=0}^3e_{si}j^s$ and $e_{si}$ are polynomials in $m,n,\gamma,p$ with integer coefficients (full expressions of $d_{ij}$ are not given). Here $\tilde{D}(j)=0$ iff $d_{1j}=d_{2j}=0$. If $\tilde{D}(j)\neq0$ for all $j$, then $x_j=y_j=z_j=0$ and $x,y,z\in\mathbb{C}$. So $\tilde{D}(j)=0$ for some $j$. Since $m,n,p\leq2$ and $\gamma p+m_z\leq m+n$, we test through all possible cases of $(m,n,\gamma,p)$. We find that only $(1,1,1,1),$ $(2,2,1,2), (2,2,3,1)$ satisfy $d_{1j}=d_{2j}=0$. When $(m,n,\gamma,p)=(2,2,3,1)$, we find that one of $\tau_{i_j}=0$. The cases $(1,1,1,1),(2,2,1,2)$ yield the same solutions $C^2+C+1=0,\lambda=2\nu,F_1=\hat{\tau}_{1}\e^{\nu t}+C,F_2=C\hat{\tau}_{1}\e^{\nu t}+1,F_3=\hat{\tau}_{1}\e^{\nu t}-1,\hat{\tau}_{1}\in\mathbb{C}$,
    \begin{equation}
    \label{l=mx-yy-zz-xeq:11}
    \begin{aligned}
    x=C^2\nu\hat{\tau}_{1}\e^{2\nu t}F_2^{-1}F_3^{-1},y=C\nu\hat{\tau}_{1}^2\e^{2\nu t}F_3^{-1}F_1^{-1},z=(C^2-C)\nu\hat{\tau}_{1}\e^{\nu t}F_3F_1^{-1}F_2^{-1}.
    \end{aligned}
    \end{equation}
If $x,y$ have min-form and max-form respectively, then $m_x=0,m_y=M_y=m+p,\ddot{x}_0=-\delta V_x,\acute{y}_0=\delta V_y$.
If $m_z\neq0$, then (\ref{l=mx-yy-zz-xeq:04.1}) and (\ref{l=mx-yy-zz-xeq:04.2}) imply $\delta M_y=\ddot{x}_0$ which yields $Cm-n=0$; a contradiction. Thus, $m_z=0$. If $m+n> \gamma p$, then $\omega_{\infty,z}=0$. Equations (\ref{l=mx-yy-zz-xeq:4.1}) and (\ref{l=mx-yy-zz-xeq:4.2}) imply $-\delta(n+p)=C\acute{y}_0$ which yields $Cm-n=0$. So, $m+n=\gamma p$ and (\ref{l=mx-yy-zz-xeq:4}) yields
\begin{equation}
\label{l=mx-yy-zz-xeq:011}
\lambda=-(\delta M_x+C\acute{y}_0)(C+1)^{-1},~\acute{z}_0=-C(\delta M_x+C\acute{y}_0)(C+1)^{-1},~\nu=-\acute{y}_0.
\end{equation}
Substituting (\ref{l=mx-yy-zz-xeq:4}) into (\ref{3x-yy-z3eq:7.2}),(\ref{3x-yy-z3eq:7.3}) yields the system (\ref{3x-yy-zz-xeq:12}) which has the determinant $D(j)$ (full expression is not given).
By going through all possible cases of $(m,n,\gamma,p)$, we find $D(j)\neq0$ for $j\leq n+p-1$. So we obtain $\acute{y}_j=\acute{z}_j=0$ for $1\leq j\leq n+p-1$ and, $y=\acute{y}_0+O(\tau^{-(n+p)}),z=\acute{z}_0+O(\tau^{-(n+p)})$. 
We recall that $I$ in (\ref{3x-yy-zz-xeq:0012}) has degree $p$ and here, $I=C\acute{y}_0+\acute{z}_0+O(\tau^{-(n+p)})$. This shows that $n+p\leq p$ and so $n=0$.

Now we discuss when $C=-1$. Assume that neither $x$ nor $y$ takes the max-form. If both $x,y$ take min-form, then (\ref{l=mx-yy-zz-xeq:4.1}),( \ref{l=mx-yy-zz-xeq:4.2}) shows that $m=n$. Since $\gamma p+m_z\leq m+n,m\neq p$, we have $m=n=2,p=1$. If only $x$ takes mid-form, then $n=p=2,m=1$. If only $y$ takes mid-form, then $m=p=1,n=2$. If $x,y$ both take the mid-form, then $m_z=0,m=n=p=1,2$. For each case, we check the ans\"atze in Lemma \ref{anslem}$(c)$ and find that one of $\hat{\tau}_{j},\breve{\tau}_{j},\check{\tau}_{j}$ vanishes, except when $m=n=p=1$. If $m=n=p=1$, we obtain: $\lambda=2\delta\hat{\tau}_{1}(\breve{\tau}_{1}-2\hat{\tau}_{1}+\check{\tau}_{1})^{-1},\nu=0,~(\hat{\tau}_{1}^2+\breve{\tau}_1\check{\tau}_{1})(\breve{\tau}_1+\check{\tau}_{1})-4\hat{\tau}_{1}\breve{\tau}_1\check{\tau}_{1}=0$, $\hat{\tau}_{1},\breve{\tau}_1,\check{\tau}_{1}\in\mathbb{C},$ and
    \begin{equation}
    \label{l=mx-yy-zz-xeq:20}
    {\begin{aligned}
    x=\frac{(\check{\tau}_{1}-\hat{\tau}_{1})\delta\tau}{(\tau-\hat{\tau}_{1})(\tau-\check{\tau}_{1})},~y=\frac{(\hat{\tau}_{1}-\breve{\tau}_{1})\delta\tau}{(\tau-\hat{\tau}_{1})(\tau-\breve{\tau}_1)},~z=\frac{-\delta(\breve{\tau}_1+\check{\tau}_{1})(\tau-\hat{\tau}_{1})^2}{(\breve{\tau}_1-2\hat{\tau}_{1}+\check{\tau}_{1})(\tau-\breve{\tau}_1)(\tau-\check{\tau}_{1})}.    
    \end{aligned}}
    \end{equation}
    Assume that only $y$ takes max-form. When $x$ takes either min-form or mid-form, we check the ans\"atze for each $m,n,p$. However, we find that one of $\hat{\tau}_{j},\breve{\tau}_{j},\check{\tau}_{j}$ vanishes. Other forms of $x,y$ can be obtained from Remark \ref{remarkminmaxmid}.

If $x,y,z$ are rational, then Lemma \ref{anslem}$(d)(i),(ii)$ shows that $V_x=V_y=0$, which yields $C=-1,\gamma=2,m=n=p\leq2$ and $d_z=-\lambda,\nu=0$. If $m=n=p=2$, checking the ans\"atze (\ref{anslemeq:2}) shows that either $\hat{t}_{j}=\check{t}_{k}$ or $\hat{t}_{j}=\breve{t}_{k}$. If $m=n=p=1$, then we obtain the solution expressing in terms of $F_1=\lambda F_3-(1+\mathrm{i}),~F_2=\lambda F_3+\mathrm{i}-1,~F_3=t-\hat{t}_1,\hat{t}_1\in\mathbb{C}$,
\begin{equation}
\label{l=mx-yy-zz-xeq:24}
{\begin{aligned}
(x,y,z)=\left(\frac{(1+\mathrm{i})}{F_3F_1},\frac{(\mathrm{i}-1)}{F_3F_2},\frac{-\lambda^2F_3^2}{F_1F_2}\right),\left(\frac{(\mathrm{i}-1)}{F_3F_2},\frac{-(1+\mathrm{i})}{F_3F_1},\frac{-\lambda^2F_3^2}{F_1F_2}\right).
\end{aligned}}    
\end{equation}

\subsection{$\mathbf{p_z},\mathbf{p_x},\mathbf{p_y},\mathbf{p_0}$ poles and $\lambda=\mu\neq\nu$}
\label{l=mx-yy-zz-x3}
The resonance conditions are $B=C,A=-1/C,C^2+\gamma C+1=0$. So $C=-1$ iff $C\in\mathbb{Q}$ which is equivalent to $\gamma=2$. As (\ref{eq:7}) becomes $k^2-k+(\gamma+1)/(\gamma+2)=0$ and has no roots in $\mathbb{N}$, then $x,y,z$ have unique expansions about any $t_0\in S_0$. Thus, Lemmas \ref{anslem}$(a)$ and \ref{subanslem} show that $x,y,z$ are non-elliptic in Class W. Since $C(C-1)V_x=C(p-n+p\gamma)+1+n+p$, we note that $V_x\neq0$. Similarly, $V_y\neq0$. Then Lemma \ref{anslem}$(c),(d)$ shows that $x,y,z$ are simply-periodic and cannot take mid-form. Ans\"atze for $x,y,z$ are given in (\ref{anslemeq:01}) and (\ref{anslemeq:001}). 

Assume that $C\neq-1$. If $x,y$ both take min-form, then $m_x=m_y=0$. If $m_z\neq0$, then (\ref{l=mx-yy-zz-xeq:04.1}) and (\ref{l=mx-yy-zz-xeq:04.2}) imply $C\ddot{y}_0=\ddot{x}_0$ which yields $Cb+(\gamma+1)(m+1)+n=0$ for some integer $b$; a contradiction. So we obtain $m_z=0$ and (\ref{l=mx-yy-zz-xeq:004}). If $m+n+1>\gamma p$, then $\omega_{\infty,z}=0$. Equations (\ref{l=mx-yy-zz-xeq:4.1}) and (\ref{l=mx-yy-zz-xeq:4.2}) yield $\delta M_x=\delta M_y$ which implies $n=m$. We recall that $I$ in (\ref{3x-yy-zz-xeq:0012}) has degree $p+1$.
If $2m+1-\gamma p\geq m+p+1$, then $I=O(\tau^{-(m+p+1)})$, which implies $m+p+1\leq p+1$ or $m=0$. So $2m+1-\gamma p< m+p+1$. Let 
\begin{equation}
\label{l=mx-yy-zz-x3eq:1}
h=f \text{~in (\ref{l=mx-yy-zz-xeq:1}) if~}\gamma=1,\text{~ otherwise~} h=g \text{~in (\ref{l=mx-yy-zz-xeqsub:1})}.   
\end{equation}
The function $h$ has only the $\mathbf{p}_0$ poles of third order with a unique Laurent expansion. Since $m(r,f)=S(r,y)$, then $h$ is a rational function in $\tau$ with degree $3$ and we have $h=h_0+O(\tau^{-(2m+1-\gamma p)})$ for some constant $h_0$. So $1\leq 2m+1-\gamma p\leq3$ or $m=(\gamma p+\ell)/2$ where $\ell=0,1,2$.
The determinant of (\ref{l=mx-yy-zz-xeq:9}) is $-(C(\gamma^2-1)+\gamma)C^{-2}(C-1)^{-2}\tilde{D}(j)$ where 
\[
\tilde{D}(j)=j^3(\gamma+2)+jd_1/4+\{(\gamma +2)( \ell+\gamma p)+2\gamma+2\}d_0/8.
\]
Here $d_1= (\gamma +2)^3(\gamma-1)p^2+2 \left(\gamma ^3-8 \gamma -8\right) p-\left(\gamma ^2+5 \gamma +6\right)\ell^2-2 (\gamma +2) \ell (\gamma+\gamma p+2p+4)-8\gamma-12$ and $d_0=(\gamma +2) (\ell^2+2\ell)-(\gamma^2-4)p(\gamma p+2p+2)+4$.
If $\gamma\geq3$, we see that $D(0)\leq0$. Since $0$ is an inflection point of $\tilde{D}(j)$, then there is at most one positive integer root of $\tilde{D}(j)$. We find that 
\[
64\tilde{D}\{(m+p)/2+1\}\leq-17p\gamma~\text{if}~\ell=0,\quad 64\tilde{D}\{(m+p)/2+1\}\leq-56p\gamma^3~\text{if}~\ell=1,2.
\]
This shows that $\tilde{D}(j)\neq0$ and $\ddot{x}_j=\ddot{y}_j=\ddot{z}_j=0$ for $j=1,\dots,\lfloor(m+p)/2+1\rfloor$. Thus, $h=\tilde{h}_0+O(\tau^{\lfloor(m+p)/2+1\rfloor+1})$
for some constant $\tilde{h}_0$, and we obtain $\lfloor(m+p)/2+1\rfloor+1\leq3$ which shows that $m=p=1$. By checking ans\"atze, we find that one of $\hat{\tau}_{j}=0$; a contradiction. When $\gamma=1$, we obtain from above that $2m+1-p\leq3$ or $2m-p\leq2$. By substituting (\ref{l=mx-yy-zz-xeq:3,2}) into (\ref{3x-yy-z3eq:7.3}), we have $\acute{z}_1=\cdots=\acute{z}_{m+p}=0$ and $z=\acute{z}_0+O(\tau^{-(m+p+1)})$. These imply $m+p+1\leq2m+1$ or $p\leq m$. Hence $2p-p\leq2m-p\leq2$ or $p\leq2$. By checking ans\"atze, we find the same contradiction. Hence, $m+n+1=\gamma p$. Arguing similarly yields $z=\acute{z}_0+O(\tau^{-(n+p+1)})$. So $I=\acute{z}_0+O(\tau^{-(n+p+1)})$ which implies $n+p+1\leq \deg I=p+1$ i.e. $n=0$. If $x,y$ take min-form and max-form respectively, then $m_x=0,m_y=M_y$. We argue similarly as in section \ref{l=mx-yy-zz-x} and find that $m_z=0,m+n+1=\gamma p$. So we also have (\ref{l=mx-yy-zz-xeq:011}) and the determinant of (\ref{3x-yy-zz-xeq:12}) is $(C+1)^{-2}C^{-1}(C+1)^{-1}D(j)$ for $j=1,\dots,n+p$, where $D(j)=C\{(\gamma+2)b_1-4-3m-3p\}+(\gamma+2)b_2+1+m+p$. Here $b_i$ are polynomials in $j,m,p,\gamma$ with integer coefficients. If $D(j)=0$, then we see that $-4-3m-3p,1+m+p$ are divisible by $\gamma+2$ which implies $1$ is also divisible by $\gamma+2$; a contradiction. So $D(j)\neq0$ and we obtain $\acute{z}_j=\acute{y}_j=0$ for $j=1,\cdots n+p$. This shows that $I=C\acute{y}_0+\acute{z}_0+O(\tau^{-(n+p+1)})$ which implies $n+p+1\leq p+1$ or $n=0$. 

Now we discuss the case $C=-1,\gamma=2$. If $x,y$ take min-form and max-form respectively, equations (\ref{l=mx-yy-zz-xeq:04.1}) and (\ref{l=mx-yy-zz-xeq:04.2}) yield $\delta(m+p+1)=\ddot{x}_0$, which implies $4p+1=0$; a contradiction. So we only need to discuss when $x,y$ both take the min-form. Then (\ref{l=mx-yy-zz-xeq:04.1}),(\ref{l=mx-yy-zz-xeq:04.2}) imply $C\ddot{y}_0=\ddot{x}_0$ which yields $m=n$. If $m_z\neq0,$ and $2m+1>2p+m_z$, then (\ref{l=mx-yy-zz-xeq:04.2}) and (\ref{l=mx-yy-zz-xeq:4.2}) imply $-\delta M_y=\lambda=\delta V_x$ which yields $3+4n=0$. Thus, $2m+1=2p+m_z$. Substituting (\ref{l=mx-yy-zz-xeq:3,2}) into (\ref{3x-yy-z3eq:7.3}) yields $\acute{z}_j=0$ for $j=1,\dots,m+p$. So $I=\acute{z}_0+O(\tau^{-(m+p+1)})$ which yields $m=0$. Hence, $m_z=0,2m+1-2p\geq1$. From (\ref{l=mx-yy-zz-xeq:4}), we have $\lambda=(m+p+1)(2m+1-2p)^{-1}\nu$.
The function $h$ in (\ref{l=mx-yy-zz-x3eq:1}) satisfies $h=O(\tau^{-(2m+1-2p)})$, then $2m+1-2p\leq\deg h=3$. So $m=p,p+1$.
We cannot use the above method here, as we could find $\tilde{D}(j)=0$ for some small $j$. However, we can find $H=a_1h^{(5)}+a_2h^{(4)}+a_3h'''+a_4h''+a_5h'+a_6h+a_7h^2+a_8hh''+a_9(h')^2+a_{10}hh''$ such that $H=H_0/(t-t_0)+H_1+\cdots$ on the $\mathbf{p_0}$ poles and not all $a_i$ are zero. Then $\tilde{H}=H_0H'+H^2-2H_1H$ is entire and has a unique expansion about any $t_0\in S_0$. Lemma \ref{lem3} shows that $\tilde{H}$ is constant. We have $\tilde{H}=\tilde{H}_0+\tilde{H}_1(t-t_0)+\tilde{H}_2(t-t_0)^2+\cdots$. Here $\tilde{H}_1,\tilde{H}_2$ are rational functions in terms of $p,\nu$ with integer coefficients. So $\tilde{H}_1=\tilde{H}_2=0$ yields two Diophantine equations in $\nu,p$. Eliminating $\nu$ gives a Diophantine equation in $p$, which we find has no positive integer roots. Hence, there are no solutions here.

\subsection{$\mathbf{p_z},\mathbf{p_x},\mathbf{p_0}$ poles and $\lambda=\mu=\nu$}
\label{l=m=nx-yy-z3}
Let $f,g$ be defined as in (\ref{neq0x-yy-z3eq:1}) and (\ref{neq0x-yy-z3eq:2}). We also have $g=(x-Cy+ACz)^2-Dzx$. If $2$ is a root of (\ref{eq:7}), then (\ref{eq:7}) implies that $\rho:=(\alpha+1)^2(\gamma-1)+C^2(\alpha-1)(\gamma+1)^2+C(\alpha+1)(\gamma+1)(\alpha\gamma+\alpha+\gamma-3)=0$ (vice versa). Here, $\alpha=\gamma=1$ implies $\rho=0$. If $\rho\neq0$, the same arguments as in section \ref{neq0x-yy-z3} shows that $g=3f_0f_2-f_1^2$. We find that $3f_0f_2-f_1^2=0$, so $g=0$. We apply (\ref{eq:22}) so that $\lambda=0$. Here, $X,Y,Z$ may have an extra singularity at $0$. Moreover, we have $h:=(X-CY+ACZ)^2-DZY=0$. Differentiating $h=0$ yields $DZ(a_1X+a_2Y+a_3Z)=0$, where $a_1=(C+\alpha)(C(2+\gamma)+1),a_2=-C(3C+1)(C+\alpha),a_3=-C(3C+\alpha)$. This implies $a_1X+a_2Y+a_3Z=0$ which contradicts Remark \ref{remarklinearatleast3} since $a_1=a_2=a_3=0$ has no roots.
If $2$ is a root of (\ref{eq:7}), then $\rho=0$.
Moreover, $x_i,y_i,z_i,$ for $i\geq2$ can be expressed in terms of $x_2$. About $t_0\in S_0$ (we assume that $t_0=0$), we have 
\begin{equation}
\label{l=m=nx-yy-z3eq:9}
{\begin{aligned}
g=\frac{3P_1x_2+\lambda^2(\alpha+1)}{P_3^2}\left(\frac{P_2}{P_1}-\frac{\lambda P_4}{2P_5}t\right)+\frac{P_1^2P_6x_2^2+\lambda^2(P_7x_2+P_8)}{8(\alpha+1)^2P_9P_5P_3^2}t^2+O(t^3),
\end{aligned}}    
\end{equation}
where $P_1=4C(\gamma+1)+4\alpha+4,$ and $P_i\in\mathbb{C}[C,\alpha,\gamma]$ (full expressions are not given). We find that $R_C(P_1,\rho),R_C(P_2,\rho),R_C(P_4,\rho)$ vanish iff $\alpha=\gamma=1$ and $R_C(P_6,\rho)=0$ iff $(\alpha-1)(\gamma-1)=0$. If $(\alpha-1)(\gamma-1)\neq0$, then $P_1P_2P_4P_6\neq0$. Thus, Lemma \ref{lemmafurtherfunction} applies here and $g\in\mathbb{C}$. If $\lambda=0$, then (\ref{l=m=nx-yy-z3eq:9}) yields $g'=0$ by Lemma \ref{lem3}. The coefficient of $t^2$ in (\ref{l=m=nx-yy-z3eq:9}) shows that $x_2=0$. Therefore, $x_i=y_i=z_i=0,$ for $i\geq1$ and $x_0^{-1}x=y_0^{-1}y=z_0^{-1}z$. Thus, $\lambda\neq0$. The coefficients of $t,1$ in $(\ref{l=m=nx-yy-z3eq:9})$ show that, by Lemma \ref{lem3}, $g'=g=0$. However, $g=0$ has been dismissed above. Hence, $(\alpha-1)(\gamma-1)=0$. The case $\alpha=\gamma=1$ has been discussed in section \ref{neq0x-yy-z}. By using (\ref{swapmaps3}), we may assume that $\alpha=1,\gamma>1$. Then, $\rho=0$ implies that $C=-1/(\gamma+1)$, and (\ref{l=m=nx-yy-z3eq:9}) becomes $g=-(\lambda^2+6x_2)(1+2\lambda t)+O(t^2)$ which shows that $g'-2\lambda g=O(t)$. So $g'-2\lambda g=0$ by Lemma \ref{lem3}. Thus, $g=Ke^{2\lambda t}$ for $K\in\mathbb{C}$. Using (\ref{eq:22}), we obtain $h=K$ ($h$ defined above). Differentiating both sides yields $Z+Y\gamma+Z\gamma=0$, which contradicts Remark \ref{remarklinearatleast3}. 

\subsection{$\mathbf{p_z},\mathbf{p_x},\mathbf{p_y}$ poles and $\lambda=\mu=\nu$}
\label{l=m=nx-yy-zz-x}
We construct an entire function of the form $x^qy^rz^s$. If $\alpha=\beta=\gamma=2$, then $D=0$. If $\alpha,\beta\geq2,\gamma>2$, then $xyz=0$ by Lemma \ref{lem3}. If $\alpha=\beta=1$, then $ABC=1,\gamma=1$, which is discussed in section \ref{neq0x-yy-z}. Using $\pi_*$, we assume that $\alpha=1,1<\beta\leq\gamma$. Consider the following relations $q-r-s=u,\gamma s-q-r=v,\beta r-q-s=w$ which yield
\begin{equation*}
\label{l=m=nx-yy-zz-xeq}
Er=(\gamma+1)u+2v+(\gamma-1)w,~Es=(\beta+1)u+(\beta-1)v+2w
\end{equation*}
where $E=(\beta-1)(\gamma-1)-4$. We set $u=Ek_1,v=Ek_2,w=Ek_3$ for $k_1,k_2,k_3\in\mathbb{N}$. Thus, $q,r,s\in\mathbb{N}$ if $E>0$, which is true unless $(\alpha,\beta,\gamma)=$ $(1,2,2),$ $(1,2,3),$ $(1,2,4),$ $(1,2,5),(1,3,3)$. Therefore, $u,v,w\geq1$, and $x^qy^rz^s$ has zeros on $S_3$. So $x^qy^rz^s=0$ by Lemma \ref{lem3}. The cases $(1,2,5),(1,3,3)$ correspond to $D=0$.
The case $\alpha=\gamma=1$ is discussed in section \ref{neq0x-yy-z}. We will go through the rest by showing that (\ref{eq:1}) is solvable and describing how to obtain the explicit solutions at the end of this section. In \cite{lvsystem90}, the authors mention the local explicit forms of $x,y,z$; however, using our procedure, we obtain the global explicit forms. We assume $\lambda=0$ via (\ref{eq:22}) and disregard the meromorphicity of $x,y,z$ here. It has been shown that for each $(\alpha,\beta,\gamma)$, there is a first integral \cite[Table II,p.689]{lvsystem90}. We denote $f=a_1x+a_2y+a_3z$ and set $a_1=-1$. 

If $(\alpha,\beta,\gamma)=(1,2,2)$, then, by using (\ref{conjugatemap}), we choose $C=-(2+\mathrm{j})/3,A=-2-\mathrm{j},B=-1-\mathrm{j},\mathrm{j}=-(1+\mathrm{i}\sqrt{3})/2$. Let $P_1(x,y,z)$ be the polynomial first integral in \cite[Table II, No.16]{lvsystem90}, then $(9\mathrm{i}-3\sqrt{3})P_1(x,y,z)=h$ for some $h\in\mathbb{C}$ (we rescale the first integral so that our calculation becomes neater). We express $(9\mathrm{i}-3\sqrt{3})P_1(x,y,z)=h$ as $a_4f''+a_5f^3+a_6h+a_7ff'=0$ and substitute $f=-x+a_2y+a_3z$ to obtain the expression of the polynomial in $x,y,z$. By letting all the coefficients vanish, we obtain a system of equations in $a_i$ to which we solve and obtain the values of $a_2,\dots,a_7$. This method still works for other cases of $\alpha,\beta,\gamma$ as well. Therefore, we obtain the expression 
    \begin{equation}
    \label{exceq:5}
    (9+3\mathrm{i}\sqrt{3})f^3+54ff'+9(3-\mathrm{i}\sqrt{3})f''-\mathrm{i}h=0,    
    \end{equation}
    where $a_2=-\mathrm{i}/\sqrt{3},a_3=-(1+\mathrm{i}\sqrt{3})/2$.
    By using $f=(3-\mathrm{i}\sqrt{3})u'/(2u)$, we find that (\ref{exceq:5}) becomes $u'''-\tilde{h}^3u=0$, where $\tilde{h}^3=(\mathrm{i}-\sqrt{3})h/108$. 

If $(\alpha,\beta,\gamma)=(1,2,3)$, then, by using (\ref{conjugatemap}), we choose $C=(\mathrm{i}-2)/5,A=(\mathrm{i}-3)/2,B=\mathrm{i}-1$. Let $P_2(x,y,z)$ be the first integral in \cite[Table II, No.17]{lvsystem90}, then we obtain
$P_2(x,y,z)=(7-24\mathrm{i})h$ for some $h\in\mathbb{C}$. Arguing similarly, we obtain the expression
    \begin{equation}
    \label{exceq:8}
    {\begin{aligned}
    a_4h+a_5f^4+a_6(f')^2+a_7f'f^2+a_8ff''+a_9f'''=0, 
    \end{aligned}}    
    \end{equation}
    where $a_2=(-1-2\mathrm{i})/5,a_3=(-1-\mathrm{i})/2,a_4=\{(1+\mathrm{i})a_7-\mathrm{i}a_8\}/16,a_5=\{(1+\mathrm{i})a_7-\mathrm{i}a_8\}/4,a_6=\{(5\mathrm{i}-5)a_7+9a_8\}/2,a_9=(1-\mathrm{i})a_8-\mathrm{i}a_7$.
    Using $f=(1-\mathrm{i})u'/u$, (\ref{exceq:8}) becomes
    \begin{subequations}
    \label{exceq:9}
    {\begin{align}
    \label{exceq:9.1}
    &u^{(4)}-\tilde{h}^4u=0,~\tilde{h}=\sqrt[4]{h}/2&\text{~if~}a_8=(2-2\mathrm{i})/3,~a_7=1,\\
    \label{exceq:9.2}
    &2u'''u'-(u'')^2-\tilde{h}^4u^2=0&\text{~if~}a_8=(1-\mathrm{i})/2,~a_7=1.    
    \end{align}}
    \end{subequations}
    These imply $u=A_1v+A_2v^{-1}+A_3w+A_4w^{-1},~A_1A_2+A_3A_4=0$ where $v=\e^{\tilde{h}t},w=\e^{\mathrm{i}\tilde{h}t},A_i\in\mathbb{C}$ (for arbitrary $a_7,a_8$, (\ref{exceq:8}) is a linear combination of (\ref{exceq:9})). 

If $(\alpha,\beta,\gamma)=(1,2,4)$, then, by using (\ref{conjugatemap}), we choose $C=(\mathrm{j}-2)/7,A=(\mathrm{j}-4)/3,B=\mathrm{j}-1,\mathrm{j}=-(1+\mathrm{i}\sqrt{3})/2$. Let $P_3(x,y,z)$ be the first integral in \cite[Table II, No.19]{lvsystem90}, then $P_3(x,y,z)=(143+180\mathrm{i}\sqrt{3})h$ for some $h\in\mathbb{C}$. Arguing similarly, we obtain the expression
    \begin{equation}
    \label{exceq:12}
    {\begin{aligned}
    &0=a_4(f'')^2+a_5f''f'f+a_6f'''f^2+a_7(f')^3+a_8(f')^2f^2+a_9f''f^3+a_{10}f'f^4\\&\textcolor{white}{0=}+a_{11}f^6+a_{12}f^{(4)}f+a_{13}f'''f'+a_{14}h+a_{15}f^{(5)},
    \end{aligned}}
    \end{equation}
    where 
    \begin{equation*}
    \begin{aligned}
    &a_2=(2\mathrm{j}^2-1)/7,~a_3=-(\mathrm{j}+2)/3,~a_4=-\{\mathrm{j}a_6+\mathrm{j}^2(21a_{15}+a_9)\}/2,\\
    &a_7=\{15a_5-11a_6+165\mathrm{j}^2a_{10}-\mathrm{j}(73a_{15}-154a_9)\}/3,~a_{11}=\{a_6+\mathrm{j}(a_9-a_{15})\}/6\\
    &a_8=\{25a_{15}-27\mathrm{j}a_{10}-31a_9-\mathrm{j}^2(2a_5+13a_6)\}/2,~a_{12}=\mathrm{j}^2(a_9-a_15)+a_{10}\\
    &a_{13}=\mathrm{j}(a_6-a_5)-\mathrm{j}^2(7a_{15}+10a_9)-11a_{10},~a_{14}=\{a_6+\mathrm{j}(a_9-a_{15})\}/162.    
    \end{aligned}    
    \end{equation*}
     Setting $f=-\mathrm{j}u'/u$ and $a_{15}=0,a_9=-\mathrm{j},a_{10}=1$ reduces (\ref{exceq:12}) to
    \begin{subequations}
    \label{exceq:13}
    {\begin{align}
    \label{exceq:13.1}
    &hu^3-216(u'')^3+324u'u''u'''-81u(u''')^2=0&\text{~if~}a_5=-\mathrm{j}^2,a_6=0,\\
    \label{exceq:13.2}
    &2(u'')^2+3u'u'''-u^{(4)}u=0&\text{~if~}a_6=\mathrm{j}^2,a_5=0.  
    \end{align}}
    \end{subequations}
    Manipulating (\ref{exceq:13}) yields $hu'+27u^{(7)}=0$. With (\ref{exceq:13}), we obtain $u=A_1\e^{\tilde{h}t}+A_2\e^{-\tilde{h}t}+A_3\e^{\mathrm{j}\tilde{h}t}+A_4\e^{-\mathrm{j}\tilde{h}t}+A_5\e^{\mathrm{j}^2\tilde{h}t}+A_6\e^{-\mathrm{j}^2\tilde{h}t}+A_7\tilde{h}^{-6},\tilde{h}^6=-h/27,A_i\in\mathbb{C}$, where all $A_i$ satisfy
    \begin{equation*}
    \label{exceq:15}
    {\begin{aligned}
    &6A_4A_6\tilde{h}^6+A_1A_7=0,~A_2A_4(36A_5A_6\tilde{h}^{12}+\mathrm{j}A_7^2)-3\mathrm{j}^2A_5^2A_6\tilde{h}A_7=0,\\
    &36A_5A_6\tilde{h}^{12}-A_7^2=0,~A_3A_4(36A_5A_6\tilde{h}^{12}+\mathrm{j}A_7^2)+\mathrm{j}^2A_5A_6A_7^2.   
    \end{aligned}}    
    \end{equation*}
    For arbitrary parameters $a_5,a_6,a_9,a_{10},a_{15}$, (\ref{exceq:12}) is a polynomial in (\ref{exceq:13}) and their derivatives with respect to $t$..

We note from the discussion above that (\ref{exceq:5}) and (\ref{exceq:8}) are linearisable.
Now we describe how to obtain the explicit forms of $x,y,z$. For the above choices of $a_1,a_2,a_3$, it can be checked that $f=A^{-1}B^{-1}y+B^{-1}z$, which yields 
\begin{equation}
\label{exceq:17}
f=-x-A^{-1}B^{-1}y+B^{-1}z,~f'=\omega_1xy,~f''=f'(\omega_2x+\omega_3y+\omega_4z)    
\end{equation}
where $\omega_i\in\mathbb{C}$. By eliminating $x,z$ in (\ref{exceq:17}), we find that $y=v_2/v_1$ where $v_1,v_2$ are polynomials in $f,f',f''$. Then $x,z$ can be found from (\ref{exceq:17}). Here, $x,y,z$ are rational functions with respect to $\e^{\theta_1 t},\e^{\theta_2 t}$ where $\theta_1,\theta_2\in\mathbb{C}$ (see the explicit forms of $u$ in each case). Hence, $x,y,z$ are meromorphic, then by inverting (\ref{eq:22}), we still obtain $x,y,z$ are meromorphic (when $\lambda\neq0$).

\subsection{$\mathbf{p_z},\mathbf{p_x},\mathbf{p_y},\mathbf{p_0}$ poles and $\lambda=\mu=\nu$}
\label{l=m=nx-yy-zz-x3}
If (\ref{eq:7}) has no positive integer roots, then we use may assume that $\lambda=0$ by using (\ref{eq:22}). Solving (\ref{eq:3}) and (\ref{eq:5}) yields $X_i=Y_i=Z_i=0,i\geq1$, and so $X_0^{-1}X=Y_0^{-1}Y=Z_0^{-1}Z$ which shows that only $\mathbf{p_0}$ poles exist. Hence, (\ref{eq:7}) has a positive integer root which has been shown in \cite{lvsystem90} that $(\alpha,\beta,\gamma)$ are the same as in section \ref{l=m=nx-yy-zz-x}.
\section{Discussion}
\label{discussion}
We have found all meromorphic solutions of the three-dimensional Lotka-Volterra system (\ref{eq:1}).
For any meromorphic solution, all poles of $x$, $y$ and $z$ are simple.  Furthermore, if any one of $x$, $y$ and $z$ has a pole at some point $t_0$, then either they all have a pole at $t_0$, in which case $t_0$ is a type ${\bf p}_0$ pole, or exactly one of $x$, $y$ or $z$ is regular at $t_0$, in which case $t_0$ is a pole of type ${\bf p}_x$, ${\bf p}_y$ or ${\bf p}_z$.  For each of these types of pole, the resonances are at $-1$,$0$ and $1$, except for type ${\bf p}_0$ poles in the case $D=ABC+1\ne0$, in which case the resonances are $-1$, $k_1$ and $k_2$, where $k_1+k_2=1$ and $k_1k_2\ne 0$.  We see that in some cases in Table \ref{tab1}, we obtain $x,y,z\in W$ (Lemma \ref{anslem}). In \cite[Th.1-2]{eremenko2005}, Eremenko shows that given an autonomous differential polynomial equation $P(w,w',\dots,$ $w^{(m)})=0$ that has the finiteness property and a unique dominant term (Definition \ref{defdiffpoly}), the meromorphic solution $w\in W$. For the L-V system, the finiteness property is not immediate due to the presence of non-negative integer resonances. 
Moreover, the single-variable ODEs of $x,y,z$ in (\ref{eq:1}) do not necessarily have a unique dominant term, so results analogous to \cite[Th.1-2]{eremenko2005} cannot be applied.

 Also, the fact that at most one of the resonances $k_1$ and $k_2$ can be positive integers in the $D\ne 0$ case, shows that generically, equations with some meromorphic solutions do not possess the Painlev\'e property, so standard Kowalevskaya-Painlev\'e-type arguments are inadequate.  

Different cases arose naturally based on which types of poles were assumed to be present.  Non-negative integer resonances gave rise to resonances conditions and a lack of uniqueness for Laurent series expansions.  Nevanlinna theory played a central role in our analysis. In particular, it enabled us to construct ``small'' functions that ultimately gave rise to lower-order differential equations characterising the meromorphic solutions.  

In \cite{halburdwang14}, the authors use the fact that the resonance parameters occur far enough apart to construct a small auxiliary function. We use this approach to construct the auxiliary functions. The difficulty here is that (\ref{eq:1}) admits four types of singularities, in contrast to the differential equation concerned in \cite{halburdwang14}, which has one type of singularity. Hence, the constructed (algebraic) functions we obtain throughout sections \ref{holosolu}-\ref{abc+1not0} could have high degree. However, showing that all these constructed functions here are small is due to the slow growth of $x,y,z$ stated in Lemma \ref{lem2}. The role of slow growth of $x,y,z$ extends further to the fact that all the small constructed functions are actually constant (Lemmas \ref{lem3} and \ref{lemmafurtherfunction}). This is where we derive the finiteness property, by examining the local series coefficients that have the resonance parameters, as discussed in section \ref{neq0x-y3} and so on. Finally, in Lemma \ref{anslem}, we show that finiteness property and slow growth property imply the Class $W$ solutions (see the end of section \ref{nevanlinna}).
\bibliographystyle{ieeetr}

%{\bibliography{bibliography.bib}}

\end{document}